\documentclass[11pt]{article}

\RequirePackage[colorlinks,citecolor=blue,urlcolor=blue, hypertexnames=false]{hyperref}
\usepackage{amsthm}
\usepackage{amsmath}
\usepackage{amsfonts}
\usepackage{amssymb}
\usepackage{apacite}
\usepackage{color}
\usepackage{enumerate}
\usepackage{framed}
\usepackage{float}
\usepackage[bottom,symbol]{footmisc}
\usepackage{graphicx} 
\usepackage{mathtools}
\usepackage{multirow} 
\usepackage[]{natbib} 
\usepackage{url} 
\usepackage{xr}  
\usepackage{algorithm}
\usepackage{algpseudocode} 
\usepackage{setspace}
\usepackage[margin=1in]{geometry}

\usepackage{multirow}
\usepackage{booktabs}

\usepackage{enumitem} 
\usepackage{subcaption} 

\def\spacingset#1{\renewcommand{\baselinestretch}%
{#1}\small\normalsize} \spacingset{1}

\newtheorem{theorem}{Theorem}
\newtheorem{prop}[theorem]{Proposition}

\theoremstyle{definition}

\newcommand{\E}{\mathbb{E}}

\newcommand{\cM}{\mathcal{M}}
\newcommand{\cS}{\mathcal{S}}
\newcommand{\N}{\mathbb{N}}
\newcommand{\cN}{\mathcal{N}}

\renewcommand{\P}{\mathbb{P}}

\newcommand{\cQ}{\mathcal{Q}}
\newcommand{\R}{\mathbb{R}}

\newcommand{\cY}{\mathcal{Y}}

\newcommand{\hP}{\widehat{P}}
\newcommand{\wpi}{\widehat{\pi}}

\title{Adaptive Importance Tempering: A flexible approach to improve computational efficiency of Metropolis Coupled Markov Chain Monte Carlo algorithms on binary spaces} 

\author{\large R. Alexander Valencia-Sanchez \textsuperscript{1}, Jeffrey S.\ Rosenthal\textsuperscript{1},\\ Yasuhiro Watanabe \textsuperscript{2}, Hirotaka Tamura\textsuperscript{3} \& Ali Sheikholeslami\textsuperscript{4}}

\footnotetext[1]{Department of Statistical Sciences, University of Toronto, Canada}
\footnotetext[2]{Fujitsu Limited, Kanagawa, Japan}
\footnotetext[3]{DXR Laboratory Inc., Kanagawa, Japan}
\footnotetext[4]{Department of Electrical and Computer Engineering, University of Toronto, Canada}

\date{} 
\begin{document}

\spacingset{1.25}  

\maketitle

\begin{abstract}
\noindent Based on the algorithm Informed Importance Tempering (IIT) proposed by \cite{IIT-Li2023} we propose an algorithm that uses an adaptive bounded balancing function. We argue why implementing parallel tempering where each replica uses a rejection free MCMC algorithm can be inefficient in high dimensional spaces and show how the proposed adaptive algorithm can overcome these computational inefficiencies. We present two equivalent versions of the adaptive algorithm $-$ A-IIT and SS-IIT $-$ and establish that both have the same limiting distribution, making either suitable for use within a parallel tempering framework. To evaluate performance, we benchmark the adaptive algorithm against several MCMC methods: IIT, Rejection free Metropolis-Hastings (RF-MH) and RF-MH using a multiplicity list. Simulation results demonstrate that Adaptive IIT identifies high-probability states more efficiently than these competing algorithms in high-dimensional binary spaces with multiple modes.
\end{abstract}

\noindent%
{\it Keywords:}  Markov chain Monte Carlo, adaptive Markov chain, computational methods, Informed proposals, Parallel Tempering.

\section{Introduction}\label{sec:intro}  

Consider a probability mass function (PMF) $\pi$ over a finite discrete state space $\cS$. Given a measurable function $g$,  we want to compute an estimate of $\mathbb{E}_{\pi}(g(X))$.  Another way to look at this problem is considering that we have a non-negative function $\pi$ defined over a finite discrete space $\cS$ and we want to obtain samples from a probability distribution which PMF is proportional to the function $\pi$ for which we don't know the normalizing constant.

A standard approach to approximate the desired value is to use a Markov Chain Monte Carlo (MCMC) algorithm which proposes to build a Markov chain $\{X_n\}$ that converges to the target distribution $\pi$,  get $N$ samples from the chain and use them to create a Monte Carlo estimate.
\newpage
An example of a MCMC algorithm is Metropolis-Hastings (MH) algorithm \citep{metropolis, hastings} which defines a way to define a Markov chain that converges to the desired target distribution by using the un-normalized function $\pi(x)$ and a proposal distribution $\mathcal{Q}(y\mid x)$. 

Rejection of some proposed states is an important part of this algorithm, in some settings it's been shown that an optimal implementation of the algorithm will reject a bit more than 75\% of the proposals \citep{roberts_1997, roberts_rosenthal_2001, optimal_scaling_distontinuous_dist_neal}. The rejection rate can also be very high for high-dimensional state spaces which results in low efficiency for the algorithm \citep{IIT-Li2023}.

These rejections might be seen as an inefficiency of the algorithm as we are performing computations that don't translate in new information for the Markov chain. To address this situation, it is possible to define a Rejection-Free version of the MH algorithm \citep{rosenthal2020jump}. This algorithm defines a new Markov chain that always accepts a move to a proposed state and allows to make inference of the target distribution $\pi$. The algorithm identifies the set of neighbor states of the current state, evaluates the target distribution on each neighbor state and modify the proposal probability with the computed information. 


This Rejection-Free MH algorithm is one instance of the class of Informed Importance Tempering (IIT) algorithms proposed by \citet{Zanella2020} as it uses information of the local neighbors of the current state to modify the proposal distribution. The advantage of IIT become clear when we have access to parallel computing. It has been shown that this class of algorithms performs well in high dimensional settings \citep{zhou2022} and that it is very simple to implement an IIT version of different MCMC algorithms \citep{IIT-Li2023}

In this work we explore how we can create an algorithm that combines IIT with Parallel Tempering (PT) and show it's effectiveness using simulations in binary multi modal state spaces. The most direct way of combining the two algorithms was already discussed by \citet[][sec. 6]{rosenthal2020jump} where they implement parallel tempering where each replica uses a rejection free Metropolis Hastings algorithm and they modify the replica swap probability to correct the bias generated by the rejection free chains. This implementation, although straightforward, increases the number of computations needed to try a replica swap in a factor proportional to the dimension of the space. 

In order to address this computational inefficiency we propose an algorithm called Adaptive IIT (A-IIT). This algorithm creates a chain with  Markovian adaptations \citep{adaptive-mcmc-roberts-rosenthal} that uses the information of the neighbors of the current state to propose a rejection free move to the next state but makes the computation of the probability of a replica swap more efficient as it doesn't need additional factors to compute  the replica swap probability.



\subsection{Overview of the paper}

In section 2 we present a brief summary of MCMC, rejection free algorithms and Informed Importance Tempering that will be useful for the definition of the adaptive algorithm.

In section 3 we present the adaptive algorithm (A-IIT) and we prove that the Markov chain defined by this algorithm converges to a probability distribution and we show how to use this adaptive Markov chain to obtain an estimator for $\mathbb{E}_{\pi}(f(X))$. 

In section 4 we briefly discuss why using a rejection-free Markov chain with parallel tempering creates additional computations for the computation of the replica swap probability and how we can avoid these extra computations without affecting the convergence of the chain.

In section 5 we present more details on the implementation of the algorithms in binary multimodal spaces. We present the second version of the algorithm: Single step IIT (SS-IIT), we explain why it's equivalent to A-IIT and how we can use both in an implementation with Parallel Tempering to reduce the number of computations needed to update the Markov chain. We present a discussion on the chosen balancing function $h(r)=\sqrt{r}$.

In section 6 we present the details of the simulations performed and show the results in binary multimodal state spaces. We compare the performance of A-IIT with other MCMC algorithms. We show that in low dimensional state spaces MH still performs better than A-IIT in terms of convergence in Total Variation Distance (TVD) and in terms of how many iterations it takes for the chains to visit the modes. We present results in a binary space of dimension 3000 where A-IIT performs better than MH and IIT, we discuss how we can attribute this improvement to both the use of the balancing function $h(r)=\sqrt{r}$ and the use of a multiplicity list to reduce the number of computations needed to update the Markov chain. Additional details of these simulations and results on state spaces of other dimension (1000, 5000, 7000) can be found in the appendix.

\section{MCMC algorithms}
To build an estimate of $\mathbb{E}_{\pi}(f(X))$ Markov Chain Monte Carlo proposes to build a Markov chain $\{X_n\}$ that converges to the target distribution $\pi$,  get $N$ samples from that distribution and then use the Monte Carlo method to estimate the desired value as follows.

$$\mathbb{E}_{\pi}(f(X)) \approx \frac{1}{N}\sum_{i=1}^N f(X_i)$$

in general we can write this estimator in the following form.
\begin{equation}\label{eq:weighted_estimator}
  \widehat{\mathbb{E}_{\pi}(f(X))}=\frac{\sum_{i=1}^N w_i(X_i)f(X_i)}{\sum_{i=1}^N w_i(X_i)}  
\end{equation}

Where $w_i(X_i)$ is the weight assigned to the state $X_i \in S$ when the Markov chain visits that state at iteration $i$.

To build this estimator we can see that we need two components: First a process to explore the state space and while exploring, a function to define the weight of each visited state.

Metropolis-Hastings algorithm, for example, assigns $w(X_i)=1$ to every state $X_i$ every time that the chain is visited. The exploration of the state space depends on the proposal distribution and the estimation of the weight depends on the rejections happening during that exploration.

If we modify the exploration process then the weight estimation also has to be modified so the chain still converges to the desired target distribution.

\subsection{Rejection free algorithms}

As stated before, an optimal implementation of Metropolis-Hastings algorithm will consistently reject proposed states. \citet{rosenthal2020jump} present a Rejection-free version of MH algorithm. Consider a Markov chain $\{X_n\}$ with transition probabilities $P(y|x)=P(X_{n+1}=y|X_{n}=x)$. Then we can build the corresponding jump Markov chain $\{J_n\}$ that has transition probabilities given by 

$$\hP(y|x)=P(J_{n+1}=y|J_n=x)=\frac{P(y|x)}{Z(x)}$$

Where $Z(x)$ represents the probability that the chain "escapes" state $x$ in a single iteration.

\begin{equation*}
    Z(X)=\sum_{y\neq x}P(y|x)=1-P(x|x)
\end{equation*}

The jump chain inherits important properties from the original chain, such as being irreducible and reversible. Moreover if the original chain's stationary distribution is $\pi$ then the stationary distribution of the jump chain is given by $\wpi(x)=Z(x)\pi(x)$.

Along with the jump chain authors define the multiplicity list $\{M_k\}$ that represents the number of iterations the original chain stayed in its current state before accepting a move to a different state. The distribution of $M_k$ given $J_k$ is equal to the distribution of $1+G$ where $G\sim \text{Geometric}(Z(x))$.

Authors present two possible ways to define the weight of a state when using a rejection free algorithm. In particular theorem 8 and 12 of \cite{rosenthal2020jump} prove that using $w_k(X_k)=M_k$ or $w_k(X_k)=1/Z(X_k)$ in equation \ref{eq:weighted_estimator} creates a consistent estimator for $\E_{\pi}(f(X))$.

The main advantage of using the multiplicity list is that it allows to recover samples of the original Markov chain. This aspect of the multiplicity list is an important part of the implementation of A-IIT that will be presented in the next section.

\subsection{Informed Importance Tempering}

 A generalization of rejection-free algorithms is presented by \citet{Zanella2020} introducing the concept of balancing function. This function measures the importance of neighbor states and uses this information to modify the way the chain explores the state space.

To build a Markov chain that follows IIT we will start with a proposal distribution $\cQ(y|x)$ that defines the set of neighbor states of a state $x$ as $\cN_x=\{y:\cQ(y|x)>0\}$. Then we consider a balancing function, which is a function $h:(0,\infty) \rightarrow (0,\infty)$ that satisfies $h(r) = rh(r^{-1})$.
For the states $y$ in the neighborhood of $x$ we consider a function $\alpha(x,y)\geq 0$ that represents the preference of choosing $y$, which then is used to modify the proposal weight of that state $\eta(y|x)$.

\begin{equation}\label{eq:weight-function}
    \eta_h(y|x)=\cQ(y|x)\alpha_h(x,y),\qquad \qquad \alpha_h(x,y)=h\left(\frac{\pi(y)\mathcal{Q}(x|y)}{\pi(x)\mathcal{Q}(y|x)}\right)
\end{equation}

Algorithm \ref{alg:IIT-weight} shows the implementation of IIT.

\begin{algorithm}
\caption{Informed importance tempering (Naive IIT) }\label{alg:IIT-weight}
\begin{algorithmic}
\State Initialize $X_0$
\For{$k$ in $0$ to $K$}
    \State For each $Y \in \mathcal{N}_{X_{k}}$, calculate $$\eta_h(Y|X_{k})=\mathcal{Q}(Y|X_{k}) h\left(\frac{\pi(Y) \mathcal{Q}(X_{k}|Y)}{\pi(X_{k}) \mathcal{Q}(Y|X_{k})}\right)$$ 
    \State Calculate $$Z_h(X_{k})=\sum_{Y\in \mathcal{N}_{X_{k}}}\eta_h(Y|X_{k})$$
    
    \State Set $w_k(X_{k}) \leftarrow 1/ Z_{h}(X_{k})$.  
    \State Choose the next state $Y\in \mathcal{N}_{X_{k}}$ such that
    $$P(X_{k+1}=Y|X_{k}) \propto\mathcal{Q}(Y|X_{k}) h\left(\frac{\pi(Y) \mathcal{Q}( X_{k}|Y)}{\pi(X_{k}) \mathcal{Q}(Y|X_{k})}\right)$$ 
\EndFor
\end{algorithmic}
\end{algorithm}

This is very similar to how rejection-free Metropolis-Hastings algorithm is implemented but we can choose to use a different function, as long as the function satisfies the balancing function equality this guarantees that the resulting chain is reversible and hence converges to the desired target distribution.

Choosing a different balancing function allows the algorithm to modify the way the Markov chain explores the state space. If we use the balancing function $h(r)=\min\{1,r\}$ then IIT is equivalent to Rejection-Free Metropolis-Hastings \citep{rosenthal2020jump}, but there's no need to use a bounded balancing function, authors mention other valid balancing functions such as $h(r)=\sqrt{r}$ (\textit{squared root}), $h(r)=1 \wedge r$ (\textit{min}), $h(r)=1 \vee r$ (\textit{max}), $h(r)=(1 \wedge re^{-c})\vee (r\wedge e^{-c})$. If we choose a different balancing function, this modifies the proposal weight of the states $\eta_{h}$ which then modifies the exploration of the state space.

In figure \ref{fig:balancing_compare} we compare how these balancing functions measure the importance of neighbor states. Consider a uniform proposal distribution $\cQ$ and two neighbor states $y_1$ and $y_2$ such that $\pi(y_1)/\pi(x) = 2$ and $\pi(y_2)/\pi(x)=3$. Using the \textit{min} balancing function both of these states have the same importance so the probability that the Markov chain chooses any of them in the next iteration is the same. However using any of the other two balancing functions the difference in probability is translated in one of the states being more likely to be chosen. In this example using \textit{max} balancing function makes $y_2$ 50\% more likely to be chosen than $y_1$ while using \textit{squared root} balancing function makes $y_2$ 22\% more likely to be chosen than $y_1$

\begin{figure}[!h]
    \centering
    \includegraphics[width=0.8\linewidth]{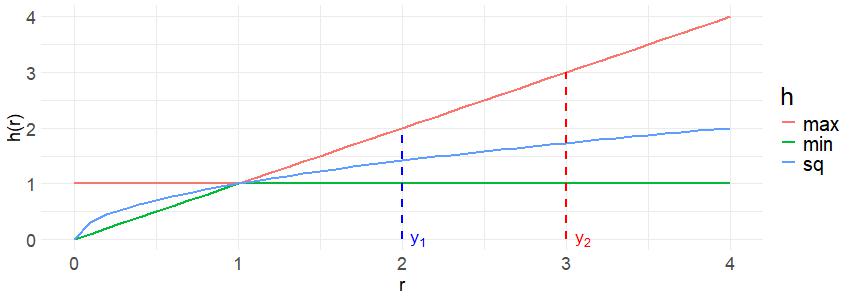}
    \caption{Comparison of balancing functions. Showing the possible values that the expression $h\left(\frac{\pi(y)}{\pi(x)}\right)$ may take for 2 neighbors of $x$: $y_1$ with lower probability and $y_2$ with higher probability.} 
    \label{fig:balancing_compare}
\end{figure}

\section{Adaptive IIT}

In general given a non-decreasing function $f(r)$ that we want to use to measure the importance of the states, we can create a version that is bounded and satisfies the balancing function property following what's presented by \citep{Zanella2020}.

For a fixed bounding constant $\gamma>1$ we first take the bounded function $f_{\gamma}(r)=\frac{1}{\gamma}\min\{\gamma,f(r)\}$ which takes values in $(0,1)$ and define:

\begin{equation}\label{eq:balancing_function_bounded}
 h_{\gamma}(r)=\min\left\{f_{\gamma}(r),r f_{\gamma}\left(\frac{1}{r}\right)\right\}   
\end{equation}

Then $h_{\gamma}$ is a balancing function taking values in $(0,1)$.

The choice of bounding constant impacts the measurement of importance. Using a small constant may result in not fully using the information of the neighbors, while using a big constant may cause numerical errors during the implementation.

For a given state space $\cS$, target probability distribution $\pi$ and non-decreasing function $f$, we define the minimal fully informative constant as $\gamma^{*} = \max\{f(\pi(y)/\pi(x));x\in \cS, y \in \cN(x)\}$ since this is the minimum constant that assigns a different weight to all states with different probabilities. Without previous knowledge of the problem, finding this constant would be unfeasible when the dimension of $\cS$ is big. To overcome this problem we propose an adaptive algorithm that updates the bounding constant as it explores the state space.

For the adaptive algorithm first we consider the basis function $f$ that will be used to measure the importance of the states. Following equation \ref{eq:balancing_function_bounded} and considering a constant $\gamma$ we can create a bounded balancing function that can be used to define the Markov chain kernel $P_{\gamma}$ as follows:

\begin{equation}\label{eq:markov_kernel_adaptive}
    \begin{split}
  P_{\gamma}(Y|X) =&\; \mathcal{Q}(Y|X)h_{\gamma}\left(\frac{\pi(Y)\mathcal{Q}(X|Y)}{\pi(X)\mathcal{Q}(Y|X)}\right) \qquad \text{for } Y\neq X, \\ 
  P_{\gamma}(X|X) =&\; 1- \sum_{y \in \mathcal{N}_{X}}  P_{\gamma}(y|X).
    \end{split}
\end{equation}

With an adequate choice of proposal distribution $\cQ$ this kernel is irreducible and aperiodic. For any fixed $\gamma>0$,  since $h_{\gamma}$ is a balancing function then the chain is reversible and $\pi(x)$ is the stationary distribution.

\begin{prop}
    Let $\pi$ be the target probability distribution, $\gamma\geq1$ be fixed and consider the Markov chain kernel $P_{\gamma}$ defined in equation \ref{eq:markov_kernel_adaptive} where $\cQ$ is a proposal distribution that makes the kernel irreducible and aperiodic and $h_{\gamma}$ is a balancing function taking values in $(0,1)$. Then $P_{\gamma}$ has $\pi$ as a stationary distribution and it's ergodic for $\pi$. 
\end{prop}

\begin{proof}
    Since $h_{\gamma}$ is a balancing function then $h_{\gamma}(r)=rh_{\gamma}(1/r)$. Using this property we obtain

    \begin{align*}
        \pi(X)P_{\gamma}(Y|X) =&\; \pi(X)\mathcal{Q}(Y|X)h_{\gamma}\left(\frac{\pi(Y)\mathcal{Q}(X|Y)}{\pi(X)\mathcal{Q}(Y|X)}\right)\\
        =&\; \pi(X)\mathcal{Q}(Y|X)\frac{\pi(Y)\mathcal{Q}(X|Y)}{\pi(X)\mathcal{Q}(Y|X)}h_{\gamma}\left(\frac{\pi(X)\mathcal{Q}(Y|X)}{\pi(Y)\mathcal{Q}(X|Y)}\right)\\
        =&\; \pi(Y)\mathcal{Q}(X|Y)h_{\gamma}\left(\frac{\pi(X)\mathcal{Q}(Y|X)}{\pi(Y)\mathcal{Q}(X|Y)}\right)\\
        =&\;\pi(Y)P_{\gamma}(X|Y)
    \end{align*}

    So $P_{\gamma}$ has $\pi$ as a stationary distribution and since the kernel is also aperiodic and irreducible over a finite state space $\cS$ then $\P_{\gamma}$ it's ergodic for $\pi$  
\end{proof}

We can then define the rejection-free kernel equivalent to equation \ref{eq:markov_kernel_adaptive} as follows:

\begin{equation}\label{eq:markov_kernel_adaptive_rf}
    \begin{split}
  \hP_{\gamma}(Y|X) =&\; \frac{1}{Z_{h}(X)}\mathcal{Q}(Y|X)h_{\gamma}\left(\frac{\pi(Y)\mathcal{Q}(X|Y)}{\pi(X)\mathcal{Q}(Y|X)}\right) \qquad \text{for } Y\neq X, \\ 
  \hP_{\gamma}(X|X) =&0.
  \end{split}
\end{equation}

Where $Z_{h}(X)$ is the probability of escaping state $X$, defined as:
\begin{equation}
    Z_{h}(X) = \sum_{Y \in \mathcal{N}_{X}}\mathcal{Q}(Y|X)h_{\gamma}\left(\frac{\pi(Y)\mathcal{Q}(X|Y)}{\pi(X)\mathcal{Q}(Y|X)}\right) 
\end{equation}

\begin{prop}
    Let $\pi(\cdot)$ be the target probability distribution, $\gamma\geq1$ be fixed and consider the rejection-free Markov chain kernel $\hP_{\gamma}$ defined in equation \ref{eq:markov_kernel_adaptive_rf} where $\cQ$ a proposal distribution that makes the kernel irreducible and aperiodic and $h_{\gamma}$ is a bounded balancing function. Then $\hP_{\gamma}$ has $\pi^{*}(\cdot)=\pi(\cdot)Z_{h}(\cdot)$ as a stationary distribution and it's ergodic for $\pi^{*}$. 
\end{prop}

\begin{proof}
    The proof is very similar to the one before and relies on the property of a balancing function.

    \begin{align*}
        \pi(X)Z_{h}(X)\hP_{\gamma}(Y|X) =&\; \pi(X)Z_{h}(X)\frac{1}{Z_{h}(X)}\mathcal{Q}(Y|X)h_{\gamma}\left(\frac{\pi(Y)\mathcal{Q}(X|Y)}{\pi(X)\mathcal{Q}(Y|X)}\right)\\
        =&\; \pi(X)\mathcal{Q}(Y|X)\frac{\pi(Y)\mathcal{Q}(X|Y)}{\pi(X)\mathcal{Q}(Y|X)}h_{\gamma}\left(\frac{\pi(X)\mathcal{Q}(Y|X)}{\pi(Y)\mathcal{Q}(X|Y)}\right)\\
        =&\; \pi(Y)\mathcal{Q}(X|Y)h_{\gamma}\left(\frac{\pi(X)\mathcal{Q}(Y|X)}{\pi(Y)\mathcal{Q}(X|Y)}\right)\\
        =&\; \pi(Y)Z_{h}(Y)\frac{1}{Z_{h}(Y)}\mathcal{Q}(X|Y)h_{\gamma}\left(\frac{\pi(X)\mathcal{Q}(Y|X)}{\pi(Y)\mathcal{Q}(X|Y)}\right)\\
        =&\;\pi(Y)Z_{h}(Y)P_{\gamma}(X|Y)
    \end{align*}

    So $\hP_{\gamma}$ has $\pi^{*}$ as a stationary distribution and since the kernel is also aperiodic and irreducible over a finite state space $\cS$ then $\hP_{\gamma}$ it's ergodic for $\pi$.
\end{proof}

Following the notation presented by \cite{adaptive-mcmc-roberts-rosenthal}, we consider a chain with Markovian adaptations $\{(X_n,\gamma_n)\}$ where at time $n$ we use the Markov kernel defined by $\gamma_n$.

To define how the kernels will be adapted, we consider a function $M:\cS \to \R$ defined as:

\begin{equation}\label{eq:bound-m}
  M(X)=\max_{Y \in \mathcal{N}(X)}\left\{f\left(\frac{\pi(Y)\mathcal{Q}(X|Y)}{\pi(X)\mathcal{Q}(Y|X)}\right),f\left(\frac{\pi(X)\mathcal{Q}(Y|X)}{\pi(Y)\mathcal{Q}(X|Y)}\right)\right\}  
\end{equation}

This function checks the possible values that the ratio of probabilities can take when evaluated using function $f$.

With this function we consider the set $\cY =\{M(X)|X\in \cS\}$ and we consider all the Markov chain kernels $P_{\gamma}$ for each $\gamma \in \cY$. 

Using these kernels the transition probabilities of the adaptive algorithm are given by 

\begin{equation}\label{eq:transition_adaptive}
    P(X_{n+1}=y|X_n=x,\Gamma_n=\gamma)=P_{\gamma}(y|x)
\end{equation}

We initialize the bounding constant at $\gamma_{0}=1$ and every time we visit a new state $X_n \in \cS$ we update the constant considering the information of the neighbor states $\gamma_{n}=\max\{M(X_n),\gamma_{n-1}\}$ and use the corresponding adapted Markov kernel to choose the next state. Algorithm \ref{alg:IIT-adaptive-rf} shows the steps to implement this adaptive Markov chain.

 \begin{algorithm}
\caption{Adaptive IIT}
\label{alg:IIT-adaptive-rf}
\begin{algorithmic}
\State Initialize $X_0$ and $\gamma_0=1$
\For{$k$ in $0$ to $K$}
\State \Comment{{\color{blue}Step 1: Update bounding constant}}
    \State For each $Y \in \mathcal{N}_{X_{k}}$, calculate $$f\left(\frac{\pi(Y)\mathcal{Q}(X_{k}|Y)}{\pi(X_{k})\mathcal{Q}(Y|X_{k})}\right),f\left(\frac{\pi(X_k)\mathcal{Q}(Y|X_K)}{\pi(Y)\mathcal{Q}(X_k|Y)}\right)$$ 
    \State Compute $M(X_{k})$ as shown in equation \ref{eq:bound-m}
    \State Update $\gamma_{k+1}=\max\{\gamma_{k}, M(X_{k})\}$
    \State Define $h_{\gamma}$ using the updated constant $\gamma_{k+1}$ following equation \ref{eq:balancing_function_bounded}
    \State \Comment{{\color{blue}Step 2: Use the adapted Markov kernel to update the chain}}
    \State Calculate $$Z_h(X_{k})=\sum_{y\in \mathcal{N}_{X_{k}}}\mathcal{Q}(Y|X_{k})h_{\gamma}\left(\frac{\pi(Y)\mathcal{Q}(X_{k}|Y)}{\pi(X_{k})\mathcal{Q}(Y|X_{k})}\right)$$
    
    \State Set $w_k(X_{k}) \leftarrow 1 + \text{Geometric}(Z_{h}(X_{k}))$.  
    \State Choose the next state $Y\in \mathcal{N}_{X_{k}}$ such that
    $$P(X_{k+1}=Y|X_{k})  = \frac{1}{Z_h(X_n)}\mathcal{Q}(Y|X_{k})h_{\gamma}\left(\frac{\pi(Y)\mathcal{Q}(X_{k}|Y)}{\pi(X_{k})\mathcal{Q}(Y|X_{k})}\right)$$ 
\EndFor
\end{algorithmic}
\end{algorithm}

Although the algorithm considers adaptations can happen every iteration, since we are considering a finite state space $\cS$, then the set of bounding constants $\cY$ is finite so this becomes a finite adaptation algorithm. Finite adaptations are one of the safe ways to create adaptive Markov chains that are ergodic \citep{adaptive-mcmc-roberts-rosenthal}.

\begin{prop}\label{prop:adaptive_ergodic}
For a function $f:(0,\infty)\to(0,\infty)$ consider the chain with Markovian adaptations $\{X_n, \gamma_n\}$ defined over a finite state space $\cS$ that follows algorithm \ref{alg:IIT-adaptive-rf} where $\cY$ is defined above and for each $\gamma \in \cY$ let $P_{\gamma}$ be the Markov chain kernel as defined in equation \ref{eq:markov_kernel_adaptive} that uses the balancing function $h_{\gamma}$ defined in equation \ref{eq:balancing_function_bounded}. Assume that for each fixed $\gamma \in \cY$ the kernel $P_{\gamma}$ is ergodic $\lim_{n \to \infty}||P_{\gamma}^{n}(x,\cdot)-\pi(\cdot)||=0$ for all $\gamma \in \cY$ and $x \in \cS$, then the adaptive algorithm is ergodic. 
\end{prop}

\begin{proof}
Let $\cM:= \max_{\gamma \in \cY}\{\gamma\}$, there exist $X' \in \cS$ such that $M(X')=\cM$. 

Using the adaptive chain we define the stopping time $\tau = \inf\{n:\N|X_n=X'\}$ which identifies the time of first visit to this state, which also corresponds to the time when the Markov chain kernel will stop changing. $\Gamma_n=\Gamma_{\tau}$ for all $n\geq \tau$.

Since the chain is irreducible and $\cS$ is finite, the states are positive recurrent then $P(\tau < \infty|X_0=x)=1$ for any $x \in \cS$. So this a finite adaptation algorithm, using Proposition 2 from \citet{adaptive-mcmc-roberts-rosenthal} we know the adaptive algorithm converges to target distribution $\pi$.
\end{proof}

Algorithm \ref{alg:IIT-adaptive-rf} considers that at every step we update the bounding constant so the proof of convergence to the target distribution relies on the fact that the state space $\cS$ is finite. If we want to extend this to an infinite discrete state space $\cS^{*}$ we will need to modify the adaptive step of the algorithm. We can consider a reduction of the bounding constant towards 1, this reduction can be random, defining each step a probability to reduce the constant, or deterministic by reducing the constant after specific number of steps. We can also consider that step 1 is run only for a deterministic number of steps $K$ and after that we use the balancing function $h_{\gamma}$ defined with the biggest found bounding constant.

Proposition \ref{prop:adaptive_ergodic} proves that the Markov chain that follows algorithm \ref{alg:IIT-adaptive-rf} is ergodic for a distribution related to the target distribution $\pi$. Using theorem 8 from \citet{rosenthal2020jump} we can see that if $\{X_k\}$ is the Markov chain constructed using A-IIT and $\{M_k\}$ is the corresponding multiplicity list then, for a suitable function $g$ the estimator:

\begin{equation}
    \frac{\sum_{k=1}^{N}M_kg(X_k)}{\sum_{k=1}^{N}M_k}
\end{equation}

is a consistent estimator of $\mathbb{E}_{\pi}(g(X))$

Adaptive IIT explores the state space $\cS$ in the same way that IIT would do it but it changes the way the weights are estimated. A-IIT can use a multiplicity list because the "escape probability" $Z_h(X)$ is bounded between $(0,1)$ which was not the case in general for IIT. Using the multiplicity list will be useful when implementing the algorithm with Parallel Tempering as it allows to define a specific number of samples to obtain from the original (non rejection-free) Markov chain.

\section{Rejection-free algorithms with Parallel Tempering}

When implementing MCMC algorithms with multi modal target distributions it's often useful to use Metropolis-Coupled chains or  Parallel Tempering \citep{geyer1991markov_parallel_tempering}.

To implement Parallel Tempering we consider $R$ independent copies (or replicas) of the chain $X$, each on the state space $\cS$ and a set of inverse temperatures $\beta_1 > \beta_2 > ... >\beta_{R} \geq0$. Replica $i$ has $\pi_i \propto \pi^{\beta_i}$ as its stationary distribution. 

Parallel tempering consist of two steps: First each replica explores the state space for a fixed number of iterations.  The second step is to propose a swap between pairs of adjacent replicas $X^{(\beta_i)}\leftrightarrow X^{(\beta_{i+1})} $. This swap is accepted using the usual Metropolis probability

\begin{equation}\label{eq:prob-swap}
    \min\left\{1,\frac{\pi_{2}(X^{(\beta_1)})\pi_{1}(X^{(\beta_2)})}{\pi_{1}(X^{(\beta_1)})\pi_{2}(X^{(\beta_2)})}\right\}
\end{equation}

This replica swap probability preserves the product target measure of the original chain $\prod_{i}\pi(\cdot)^{(\beta_i)}$.

If we would like to implement parallel tempering with each replica following a rejection-free algorithm then the swap probability needs to be adjusted \citep{rosenthal2020jump} considering that the rejection-free chain have a different target distribution. The swap probability needs to be modified as follows:

\begin{equation}\label{eq:prob-swap-adaptive}
    \min\left\{1,\frac{{\color{red}Z_2(X^{(\beta_1)})}\pi_{2}(X^{(\beta_1)}){\color{red}Z_1(X^{(\beta_2)})}\pi_{1}(X^{(\beta_2)})}{{\color{red}Z_1(X^{(\beta_1)})}\pi_{1}(X^{(\beta_1)}){\color{red}Z_2(X^{(\beta_2)})}\pi_{2}(X^{(\beta_2)})}\right\}
\end{equation}

This preserves the modified product measure $\prod_{i}Z(\cdot)\pi(\cdot)^{(\beta_!)}$ which corresponds to the target distribution where each replica is a rejection-free chain. Not adjusting the probability will induce bias in the chain that will affect the rate of convergence.

As stated before each of these $Z$ factors requires to evaluate $\pi$ for every neighbor of the state and apply the corresponding balancing function. For high dimension spaces $\cN(X)$ can become very large and although rejection-free chains always move to a different state, the replica swap may be rejected and the additional computations won't translate in new information for the Markov chain.

To avoid these additional computations we use a rejection-free algorithm with multiplicity list so we can alternate between the independent replica updates and the replica swaps without inducing bias. As explained by \citet{rosenthal2020jump}, we define a number of samples $L_0$ to obtain from the original chain (with target distribution $\pi$). Then we proceed as follows.

\begin{enumerate}
\setlength\itemsep{0.01em}
    \item Set the number of remaining repetitions to $L=L_0$
    \item Find the next state for replica $i$, $X_{n+1}^{(i)}$ and the value of the multiplicity list $M_{n}^{(i)}$ corresponding to the current state $X_{n}^{(i)}$
    \item If $M_{n}^{(i)}\geq L$ then replace $M_{n}^{(i)}$ by $L$ and the chain stays in the current state $X_{n+1}^{(i)}=X_{n}^{(i)}$. Then return to step 1  for the next replica.
    \item If $M_{n}^{(i)}< L$ then update $L=L-M_{n}^{(i)}$ and return to step 2 for the same replica
\end{enumerate}

Using a rejection-free chain with a multiplicity list allows us to transform each single rejection-free iteration into $M_n$ samples from the original chain. So we can use the steps described above to obtain $L_0$ samples from the original chain using a multiplicity list in all the replicas. This then allows us to use parallel tempering without creating bias and without additional computations when trying a replica swap. We can see it's straightforward to implement Parallel Tempering where each replica uses A-IIT following the four steps above. This is more efficient than using IIT in the replicas as we avoid the computation of the factors ${\color{red}Z_i(X^{(\beta_j)})}$ every time we try a replica swap.

 \section{Implementation of Adaptive IIT in Parallel Tempering}

\subsection{Single step IIT}
Algorithm \ref{alg:IIT-adaptive-rf} requires a lot of computations for each iteration. First the update of the bounding constant requires the computation of $\pi(\cdot)$ for each of the $\cN(X)$ neighbors of the current state $X$. Then we need to apply the bounded balancing function to each of the probability ratios to compute the weight $Z_h(X_k)$ and then choose one of the proposed states.


When implementing A-IIT with parallel tempering it's possible to use a non-rejection free algorithm for some of the replicas as long as the target distribution is not modified. To achieve this, we can consider an algorithm that uses a bounded balancing function $h$ and is equivalent to Adaptive IIT but creates just one sample of the original chain every iteration.

 \begin{algorithm}
\caption{Single-step IIT (SS-IIT)}
\label{alg:IIT-single-step}
\begin{algorithmic}
\State Initialize $X_0$ and $\gamma_0=1$
\For{$k$ in $0$ to $K$}
    \State \Comment{{\color{blue}Step 1: Update bounding constant}}
    \State Propose a new state $Y \in \cN(X_k)$ using the proposal distribution $\mathcal{Q}(\cdot|X_{k})$
    \State Calculate $$B_k=\max\left\{f\left(\frac{\pi(Y)\mathcal{Q}(X_{k}|Y)}{\pi(X_{k})\mathcal{Q}(Y|X_{k})}\right),f\left(\frac{\pi(X_k)\mathcal{Q}(Y|X_k)}{\pi(Y)\mathcal{Q}(X_k|Y)}\right)\right\}$$
    \State Update $\gamma_{k+1}=\max\{\gamma_k,B_k\}$
    \State Define $h_{\gamma}$ using the updated constant $\gamma_{k+1}$ following equation $\ref{eq:balancing_function_bounded}$ 
    \State  \Comment{{\color{blue}Step 2: Use the adapted Markov kernel to update the chain}}
    \State Set $w_k(X_{k}) \leftarrow 1$.
    \State let $U\sim Unif(0,1)$
    \If{$U<h_{\gamma}\left(\frac{\pi(Y)\mathcal{Q}(X_{k}|Y)}{\pi(X_{k})\mathcal{Q}(Y|X_{k})}\right)$}
    \State $X_{k+1} \gets Y$
    \Else 
    \State $X_{k+1} \gets X_k$
    \EndIf
\EndFor
\end{algorithmic}
\end{algorithm}


As discussed in the previous section, when implementing parallel tempering with Adaptive IIT we define a number of samples of the original chain $L_0$ to obtain before trying a replica swap. We can also use Single-step IIT to obtain the desired number of samples.

The advantage of A-IIT is that, since it's rejection free, after every iteration the chain moves to a different state however this requires additional computations. The advantage of SS-IIT is that it doesn't require as many computations to propose a jump however we may reject the proposed state and the chain may not explore the state space as fast as A-IIT.

When implementing Parallel Tempering we can choose some replicas to use A-IIT and others to use SS-IIT. For some replicas with high values of $\beta$ we expect more rejections whenever the chain moves to a state of high probability so we would prefer A-IIT. For a value of $\beta$ close to zero we don't expect as many rejections so SS-IIT might be a better choice. 

\subsection{Choosing between A-IIT and SS-IIT}

Now we have two adaptive algorithms that can generate samples from the target distribution $\pi$ while using any balancing function $h$. To choose which of them to use for the different replicas in parallel tempering we will compare the number of computations needed for the chain to generate $L_0$ samples from the original chain. In this situation we consider that each evaluation of $\pi$ in a neighbor state counts as one computation.  SS-IIT always generates 1 sample of the original chain for each iteration of the algorithm which requires to evaluate $\pi$ in the proposed neighbor state. A-IIT generates one sample from the rejection-free chain which is equivalent to $M_n$ (a random number) of samples from the original chain but it requires to evaluate $\pi$ for each state in the set $\cN_X$.

Consider the state space $\cS=\{0,1\}^p$ with $p=5000$ and a target distribution with 6 modes defined in a similar way than the target distributions presented in the next section i.e., $\pi(X)\propto \sum_{i=1}^6\exp(-\theta \;||X-X_{(i)}||_{1})$.

We run parallel tempering with 13 replicas $(X^{(\beta_1)},\dots,X^{(\beta_{13})})$. Each replica uses Adaptive IIT until the multiplicity list reaches a value of $L_0=800$ then tries a replica swap. We run 25 simulations and in each simulation each replica performs 100 swaps before stopping. We then compute the average number of rejection free steps each replica performed in-between replica swaps.

Table \ref{tab:temperatures_iteration_compare} contains the values of temperatures used and the average swap rate between $\beta_i$ and $\beta_{i+1}$.

\begin{table}[!htbp] \centering 
  \caption{Inverse temperature used} 
  \label{tab:temperatures_iteration_compare} 
\begin{tabular}{@{\extracolsep{5pt}} ccc} 
\\[-1.8ex]\hline 
\hline \\[-1.8ex] 
\multirow{2}{*}{id} & inverse & \multirow{2}{*}{swap rate}\\ 
 & temperature & \\ 
\hline \\[-1.8ex] 
$\beta_{1}$  & 20000 & 0.412\\ 
$\beta_{2}$  & 19820 & 0.382\\ 
$\beta_{3}$  & 19539 & 0.426\\ 
$\beta_{4}$  & 19409 & 0.401\\ 
$\beta_{5}$  & 19282 & 0.384\\ 
$\beta_{6}$  & 19169 & 0.414\\ 
$\beta_{7}$  & 19112 & 0.333\\ 
$\beta_{8}$  & 18935 & 0.272\\ 
$\beta_{9}$  & 18724 & 0.341\\ 
$\beta_{10}$ & 18663 & 0.318\\ 
$\beta_{11}$ & 18586 & 0.256\\ 
$\beta_{12}$ & 18507 & 0.297\\ 
$\beta_{13}$ & 18485 & - \\ 
\hline \\[-1.8ex] 
\end{tabular} 
\end{table}

In figure \ref{fig:compare-iterations} we plot the number of rejection free iterations each replica with inverse temperature $\beta_i$ needs so that their multiplicity list reaches the value $L_0=800$, i.e., before we try a replica swap. 

We can observe that as the value of $\beta$ decreases, the number of iterations increases as well. 

\begin{figure}[!h]
    \centering
    \includegraphics[width=0.8\linewidth]{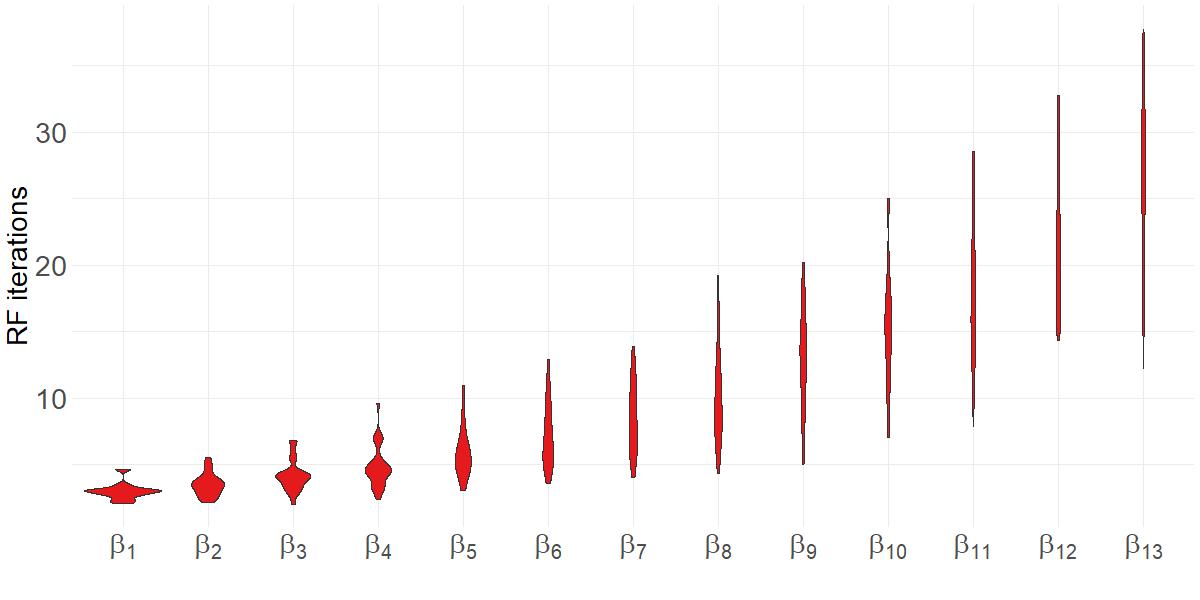}
    \caption{Number of Rejection Free iterations for replicas tempered at different temperatures} 
    \label{fig:compare-iterations}
\end{figure}

For $\beta_1$ on average the chain performs 3 rejection-free steps before the replica swap. This means that single-step IIT algorithm would need, on average, 267 rejections before accepting a move. On the other hand for $\beta_{13}$ the chain performs, on average, 26 rejection-free steps before the replica swap, single-step IIT algorithm would need, on average, 30 rejections before accepting a move.

The size of the neighbor set is the dimension of the state space $|\cN(X)|=5000$, this means that each rejection-free iteration needs to evaluate $\pi$ $5000$ times before moving to a new state no matter the value of $\beta$, meanwhile the number of computations needed for SS-IIT to move to a new state depends on the value of $\beta$, table \ref{tab:compare-iterations} shows a summary of the \% of computations needed for the chain to move to a different state for SS-IIT in comparison with A-IIT.

 \begin{table}[h]
 
\centering
\begin{tabular}{lcc}
\hline
Algorithm & $\beta_1$ & $\beta_{13}$ \\
\hline
A-IIT & 5000&5000 \\
SS-IIT & 267 &30 \\

\hline\hline
\% &5.3\% &0.6\% \\
\hline
\end{tabular}
\caption{Comparison of number of computations needed for the Markov chain to move to a different state}
\label{tab:compare-iterations}
\end{table}

In this example the the number of computations needed for the smallest $\beta$ is 10 times bigger than the computations for the largest $\beta$. When implementing Adaptive IIT with parallel tempering we can consider running a first simulation to measure the average number of rejection-free steps that algorithm \ref{alg:IIT-adaptive-rf} would need in-between replica swaps and compare with the number of computations needed. With this information we can define which algorithm to use for each replica in such a way that we reduce the number of computations needed for a replica to visit a new state and, in this way, implement an algorithm that efficiently explores the state space.

\subsection{General discussion of the balancing function}
The proposed adaptive algorithms are defined in general and can be applied for any increasing function $f$. In their work \cite{zhou2022} gives particular emphasis to the balancing function $h(r)=\sqrt{r}$ so in this work we implement a bounded version of this function to use in the adaptive algorithm. In this section we present the process to create a bounded version of this balancing function.

For a fixed $\gamma>0$ taking the function $f(r)=\sqrt{r}$ as a starting point, the bounded balancing function version of $f$ takes the following form:

\begin{equation}\label{eq:bounded_sq_balancing_func}
  h_{\gamma}(r)=\min\left\{1,r,\frac{\sqrt{r}}{\gamma}\right\}=\begin{cases}
    r \text { for } r\in\left[0,\frac{1}{\gamma^2}\right)\\
    \frac{\sqrt{r}}{\gamma} \text { for } r\in\left[\frac{1}{\gamma^2},\gamma^2\right)\\
    1 \text { for } r\in\left[\gamma^2,\infty\right)\\
\end{cases}  
\end{equation}

Figure \ref{fig:bounded-sqrt} shows the behaviour of this bounded balancing function for some values of $\gamma$.
\begin{figure}[!h]
    \centering
    \includegraphics[width=0.8\linewidth]{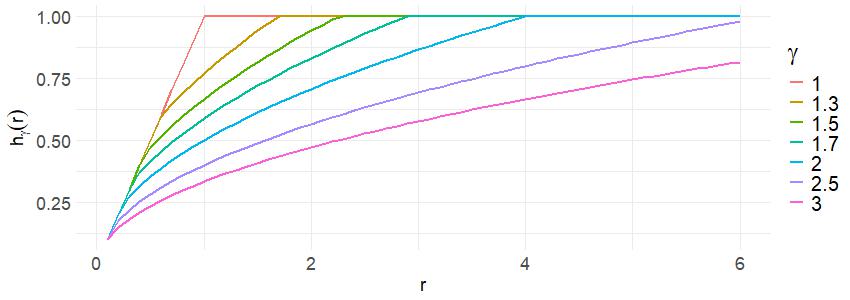}
    \caption{Bounded balancing function based on $h(r)=\sqrt{r}$ considering different values for the bounding constant $\gamma$.} 
    \label{fig:bounded-sqrt}
\end{figure}

Since we want to use $h_{\gamma}$ to compute probabilities there's no usage of using a bound smaller than 1. We can observe that for $\gamma=1$ $h_{1}(r)=\min\{1,r\}$ the balancing function is the same as the function used in Metropolis-Hastings algorithm.

The definition of algorithm \ref{alg:IIT-adaptive-rf} requires a double visit to the neighbor states. First we compute the ratio of probabilities of the neighbor states to update the bounding constant. Then, we apply the updated balancing balancing function to the target distribution evaluated at each neighbor state. This ensures that we're always using the latest information available (the latest bounding constant found). 

To avoid this double visit to the neighbor states, we could swap steps 1 and 2 in the algorithm. So first we compute the ratio of probabilities $\pi(y)/\pi(x)$ for all the neighbor states, store these values, then apply the bounded balancing function (using a previously found bound) to compute the factor $Z_h(x)$. After this computation we can use the stored values of the ratio of probabilities, apply the function $f$ and compute the new balancing function.

Changing the order of steps 1 and 2 in the algorithm helps to avoid the double visit to the neighbor states but we need extra storage capacity. If we use the squared root balancing function we can avoid this double visit in a more straightforward way.

Note that the definition of the bounded squared root balancing function (\ref{eq:balancing_function_bounded}), for a fixed $\gamma>0$ $h_{\gamma}(r)=\sqrt{r}/\gamma$ as long as $r \in [1/\gamma^2, \gamma^2]$. So we can simplify the definition of the function as long as the input is between these two values.

For simplicity consider that we're using a uniform proposal distribution $\cQ$ and consider $R_i=\pi(Y_i)/\pi(X)$. Note that the way that we update the bounding constant guarantees that $\gamma \geq R_i^{1/2}$ and $\gamma \geq R_i^{-1/2}$, this implies that $R_i \in [1/\gamma^2, \gamma^2]$ for all the neighbors of $X$, and this is ensured in every iteration of the algorithm.

Then we don't need to define a bounded balancing function to apply to each $\sqrt{R_i}$, to apply the bounded balancing function it suffices to divide all the values by $\gamma$ i.e., $h_{\gamma}(R)=\sqrt{R}/\gamma$. Algorithm \ref{alg:IIT-adaptive-efficient} shows how we can implement this more efficient version of Adaptive IIT when we use the balancing function $h(r)=\sqrt(r)$. This doesn't change the way the chain behaves but reduces the number of computations in each step of the algorithm.

 \begin{algorithm}
\caption{Adaptive IIT (computationally efficient for $f(r)=\sqrt{r}$)}
\label{alg:IIT-adaptive-efficient}
\begin{algorithmic}
\State Initialize $X_0$ and $\gamma_0=1$.
\State Consider a uniform proposal distribution $\cQ(Y|X)\propto 1$
\For{$k$ in $1$ to $K$}
    \State For each $Y_i \in \mathcal{N}_{X_{k-1}}$, calculate $$R_i(X_{k-1})^{1/2}= \left(\frac{\pi(Y_i)}{\pi(X_{k-1})}\right)^{1/2}$$ 
    \State $\gamma_{c}=\max_{i}\{R_i^{1/2},R_i^{-1/2} \}$
    \State $\gamma_{k}=\max\{\gamma_{k-1}, \gamma_{c}\}$
    \State Calculate $$Z_h(X_{k-1})=\frac{1}{|\cN(X_{k-1})|}\sum_{y\in \mathcal{N}_{X_{k-1}}}\frac{1}{\gamma_{k}}\left(\frac{\pi(Y_i)}{\pi(X_{k-1})}\right)^{1/2}$$
    
    \State Set $w_k(X_{k-1}) \leftarrow 1 + Geometric(Z_{h}(X_{k-1}))$.  
    \State Choose the next state $Y\in \mathcal{N}_{X_{k-1}}$ such that
    $$P(X_k=Y|X_{k-1}) \propto\left(\frac{\pi(Y) }{\pi(X_{k-1}) }\right)^{1/2}$$ 
\EndFor
\end{algorithmic}
\end{algorithm}

\section{Simulations}

\subsection{General framework}

For the simulations we consider a state space $\cS=\{0,1\}^p$ and a multimodal target distribution $\pi$ defined as

\begin{equation}\label{eq:target-pi}    
\pi(X)=\frac{1}{\mathfrak{C}_{\theta}}\sum_{i=1}^m\exp(-\theta \;||X-X_{(i)}||_{1})
\end{equation}

Where $m$ is the number of modes, $X_{(i)}$ represents mode $i$, $\theta>0$ and $||\cdot||_{i}$ is the $L_1$ distance, $\mathfrak{C}_{\theta}$ is the normalizing constant.

We implement non-reversible Parallel Tempering with $R$ replicas \citep{Syed2021NonReversible} alternatinv between swapping even then odd indexed replicas. We compare the performance of Adaptive IIT with other 3 algorithms. All of them use non-reversible parallel tempering but the replicas in each of the implementation will use one of 4 different algorithms. For all of the algorithms we use a uniform proposal distribution ($\cQ(y|x) \propto 1$ for all $y\in \cN_X$) and the neighbor set is defined by all the states with $L_1$ distance 1: $\cN_x=\{y\in\cS : ||x-y||_1=1\}$.




\begin{table}[h]
\centering
\setlength{\lightrulewidth}{0.01em}  
\begin{tabular}{clcccl}
\toprule
\# & ID & Name & $h(r)$ & $w(x)$ & Algorithm \\ 
\midrule
\multirow{2}{*}{1} & \multirow{2}{*}{A-IIT} & \multirow{2}{*}{Adaptive IIT} & \multirow{2}{*}{$\sqrt{r}$} & $1+G(Z(x))$ & \ref{alg:IIT-adaptive-rf} \\ 
\cmidrule(lr){5-6}
 &  &  &  & $1$ & \ref{alg:IIT-single-step} \\ 
\midrule
2 & IIT & Naive IIT & $\sqrt{r}$ & $1/Z(x)$ & \ref{alg:IIT-weight} \\ 
\midrule
\multirow{2}{*}{3} & \multirow{2}{*}{MH-mult} & \multirow{2}{*}{\begin{tabular}[c]{@{}l@{}}Metropolis Hastings\\ with multiplicity list\end{tabular}} & \multirow{2}{*}{$\min\{1,r\}$} & $1+G(Z(x))$ & \ref{alg:rejection-free-metropolis-hastings-mult} \\ 
\cmidrule(lr){5-6}
 &  &  &  & $1$ & \ref{alg:metropolis-hastings} \\ 
\midrule
4 & RF-MH & \begin{tabular}[c]{@{}l@{}}Rejection Free\\ Metropolis Hastings\end{tabular} & $\min\{1,r\}$ & $1/Z(x)$ & \ref{alg:rejection-free-metropolis-hastings-weight} \\ 
\bottomrule
\end{tabular}
\caption{Summary of compared algorithms.}
\label{tab:algorithms}
\end{table}

In table \ref{tab:algorithms} is the summary of the algorithms compared. IIT and RF-MH use direct weight estimation so they require additional computations for the replica swap probability. The algorithms that we label as A-IIT and MH-mult are implementation of 2 different algorithms. For some of the replicas with high values of $\beta$ we use a rejection free algorithm with a multiplicity list (algorithms \ref{alg:IIT-adaptive-rf} and \ref{alg:rejection-free-metropolis-hastings-mult}) while for replicas with low values of $\beta$ we use algorithms that allow rejections (algorithms \ref{alg:IIT-single-step} and \ref{alg:metropolis-hastings}), the number of replicas using a rejection-free algorithm is the same for both algorithms in any example.

Traditionally one should run the algorithms for the same amount of iterations to make a fair comparison of the performance of each of them. For this work is not straightforward to define the number of iterations to run each algorithm since algorithms (2) and (4) perform a fixed number of rejection free iterations while algorithms (1) and (3) perform a random number of rejection free iterations since they use a multiplicity list.

To define equivalent number of iterations in between replica swaps for all algorithms, we first define a number of rejection free iterations for algorithms (2) and (4) and then, for algorithms (1) and (3), we define a value of $L_0$ (the number of samples from the original chain to obtain before a replica swap) in such a way that, for the replica with the highest $\beta$, the average number of Rejection-Free steps in between replica swaps is as close as possible for all the algorithms. In this way the replicas get the same number of steps to explore the state space and the comparison is more fair.

For example for algorithms (2) and (4) we define 2 rejection-free iteration to try between replica swaps and to run the algorithms for a total of 2 million iterations, this translates in 1 million of swaps tried.
For algorithms (1) and (3) defining $L_0=800$ equals to 1.64 and 1.2 average rejection free steps. So we use that $L_0$ and the algorithms tries a total of 1 million replica swaps. In this way we try to make the  comparison as fair as possible.

The number of replicas and the distance between the $\beta$'s is defined such that the proportion of accepted swap rates between adjacent replicas is close to 0.234 \citep{atchade_2011_optimal_scaling_pt}.

We measure the time it takes the replica with the highest $\beta$ to reach all of the modes as we consider this is a measurement of how good is the chain's exploration capabilities to find states with high probability. This is presented with graphs simulating a ladder where in the x-axis we plot the time and the y-axis we plot the percentage of simulations that found all the modes of the target distribution at the defined time. Each jump in the plot represents one simulation finding all the modes of the target distribution, the faster the function reaches 100\% the faster the algorithm is in finding all the modes.

For the low dimensional examples we also measure the evolution of Total Variation Distance (TVD). We present this in a chart where the x-axis is the time and the y-axis is the average TVD calculated of all the simulations ran. The faster the plot approaches $0$ the better the algorithm is at converging to the desired target distribution.

The simulation results for the \textbf{low dimensional} examples show that the rejection free algorithms perform better both in terms of TVD and in speed to find the modes. This is expected as the  rejection free algorithms efficiently use the information of neighbor states to identify direct paths to the modes. In this low dimensional example the additional computations needed by the algorithms using a multiplicity list (A-IIT and MH-mult) is a lot more than the computations needed by algorithms using direct weight estimation (IIT and RF-MH). This is due to the dimension of the state space being small and $L_0$ being large in comparison. So the number of computations in-between replica swaps for algorithms with multiplicity list is always larger than the computations for the other two algorithms and the additional computations needed during the replica swap process is not significant.


It is in the simulation results for the \textbf{high dimensional} example that we observe A-IIT performing better than the other algorithms. Since the dimension of the space is $p=3000$ the computation of the probability of swapping adjacent replicas requires an additional $4\times 3000=12000$ evaluation of the function $\pi$, and since we are performing non-reversible replica swaps each step of replica swaps requires $12000 \times 6=72000$ additional computations, this impacts the speed in which the parallel tempering algorithm explores the state space. 

Additional to that we observe that the value of $L_0$ does not scale at the same rate as the dimension of the state space so, for the algorithms using multiplicity list, we are able to find an adequate choice of replicas that use the equivalent algorithm that permits rejections in such a way that we minimize the number of computations needed per replica swap.

In the simulation results, we can also see the difference in speed between A-IIT and MH-mult which can be attributed to using a different balancing function. Hence we observe the advantage of using A-IIT over Metropolis Hastings or IIT when using parallel tempering.

\subsection{Low dimensional bi-modal example}
For this example we consider a space of dimension $p=16$ with two modes:
\begin{itemize}
\setlength\itemsep{-0.5em}
    \item $x_{(1)}=(1,0,1,0,1,0,1,0,1,0,1,0,1,0,1,0)$
    \item $x_{(2)}=(0,1,0,1,0,1,0,1,0,1,0,1,0,1,0,1)$
\end{itemize}



The state space is $\cS=\{0,1\}^{16}$, we use $\theta=6$ to define the target distribution, with these parameters, the normalizing constant is $\mathfrak{C}_{\theta}=2(1+e^{-\theta})^{16}$, each mode accumulates $48\%$ of the total probability mass. Since this space is small enough we can compute the true distribution $\pi$ and use it to compute the Total Variation Distance.

We run all algorithms for a total of 1 million replica swaps, the algorithms using rejection free iterations try a replica swap after every iteration. Algorithm A-IIT uses $L_0=1000$ and algorithm MH-mult uses $L_0=100$. We run 100 simulations of each algorithm to obtain the results presented below.

For this simulations, all the replicas of A-IIT and MH-mult use the rejection free algorithm. In table \ref{tab:swap_rate_bimodal} we see the average of rejection free swaps that each replica performs in-between replica swaps with the defined $L_0$.

\begin{table}[!htbp] \centering 
  \caption{Average number of iterations between replica swaps} 
  \label{tab:avg_iter_bimodal} 
\begin{tabular}{@{\extracolsep{5pt}} ccc} 
\\[-1.8ex]\hline 
\hline \\[-1.8ex] 
replica & A-IIT & MH-mult \\ 
\hline \\[-1.8ex] 
$\beta_{1}$ & 1.246 & 1.489 \\ 
$\beta_{2}$ & 24.067 & 10.942 \\ 
$\beta_{3}$ & 91.075 & 25.027 \\ 
$\beta_{4}$ & 218.88 & 42.84 \\ 
\hline \\[-1.8ex] 
\end{tabular} 
\end{table}

In this particular problem we don't have benefits of using A-IIT as using IIT is faster both in terms of finding the modes and convergence to the target distribution. We attribute this to the small size of the state space, the number of additional computations needed by IIT both for each iteration and trying replica swaps does not hinder the performance.



\begin{figure}[htpb]
    \centering
    \includegraphics[width=0.95\linewidth]{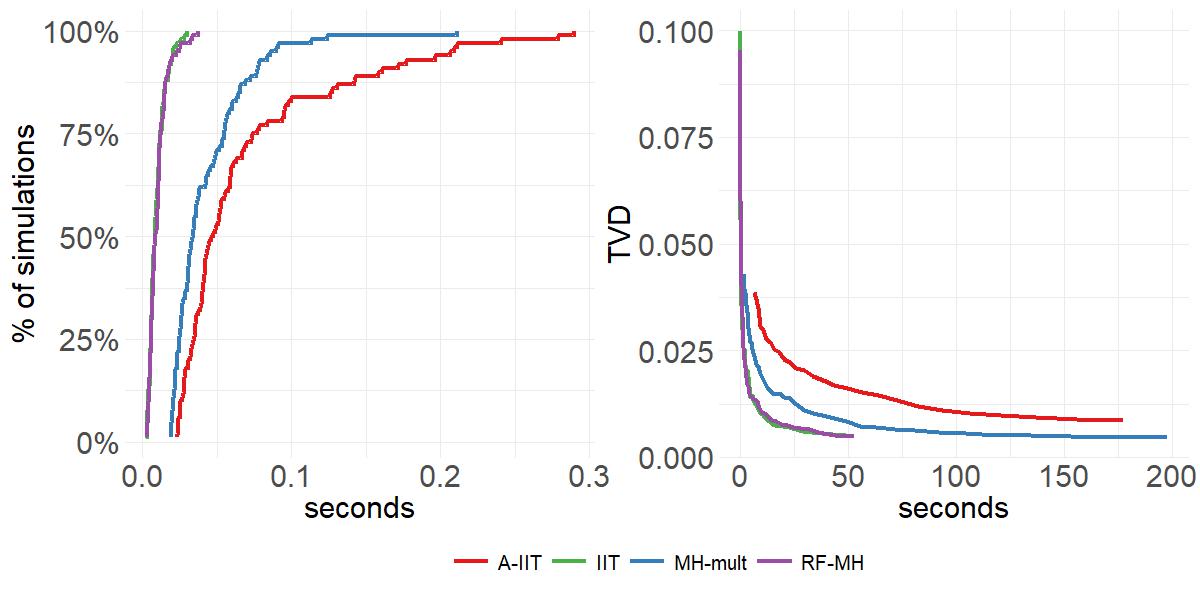}
    \caption{Comparison of performance of the four algorithms in the low dimensional bimodal problems. On the left, time to visit all the modes of the target distribution. On the right, evolution of Total Variation Distance} 
    \label{fig:lowd_bimodal_dual}
\end{figure}

\newpage
\subsection{Low dimensional multi modal example}

For this example we consider a space of dimension $p=16$ with seven modes:
\begin{itemize}
\setlength\itemsep{-0.5em}
    \item $x_{(1)}=(1,1,1,1,1,1,1,1,1,1,1,1,1,1,1,1)$
    \item $x_{(2)}=(1,0,1,0,1,0,1,0,1,0,1,0,1,0,1,0)$
    \item $x_{(3)}=(0,1,0,1,0,1,0,1,0,1,0,1,0,1,0,1)$
    \item $x_{(4)}=(1,1,1,1,1,1,1,1,0,0,0,0,0,0,0,0)$
    \item $x_{(5)}=(0,0,0,0,0,0,0,0,1,1,1,1,1,1,1,1)$
    \item $x_{(6)}=(1,0,0,0,0,0,0,0,0,0,0,0,0,0,0,1)$
    \item $x_{(7)}=(0,0,0,0,0,0,0,1,1,0,0,0,0,0,0,0)$
\end{itemize}

The state space is $\cS=\{0,1\}^{16}$, we use $\theta=10$ to define the target distribution, with these parameters, the normalizing constant is $\mathfrak{C}_{\theta}=7(1+e^{-\theta})^{16}$, all modes accumulate 99.9\% of the total probability mass. Since this space is small enough we can compute the true distribution $\pi$ and use it to compute the Total Variation Distance.

We run all algorithms for a total of 1 million replica swaps, the algorithms using rejection free iterations try a replica swap after every 2 iterations while the algorithms using a multiplicity list uses an $L_0=1000$. We run 100 simulations of each algorithm to obtain the results presented below.

For this simulations, all the replicas of A-IIT and MH-mult use the rejection free algorithm. In table \ref{tab:swap_rate_7_mode} we see the average of rejection free swaps that each replica performs in-between replica swaps with the defined $L_0$. Here we observe the effect of the adaptive constant. A bounding constant bigger than $1$ means that the probability of escape $Z_h(X)$ is smaller which then translates in bigger values for the multiplicity list, reducing the number of rejection-free iterations needed to reach the defined $L_0$.

\begin{table}[!htbp] \centering 
  \caption{Average number of iterations between replica swaps} 
  \label{tab:avg_iter_7_mode} 
\begin{tabular}{@{\extracolsep{5pt}} ccc} 
\\[-1.8ex]\hline 
\hline \\[-1.8ex] 
replica & A-IIT & MH-mult \\ 
\hline \\[-1.8ex] 
$\beta_{1}$ & 1.001 & 1.091 \\ 
$\beta_{2}$ & 19.306 & 87.393 \\ 
$\beta_{3}$ & 77.602 & 221.185 \\ 
\hline \\[-1.8ex] 
\end{tabular} 
\end{table}

We observe a similar behavior as in the bi-modal problem. Using parallel tempering with IIT gives a faster algorithm to both find all the modes and estimating the the target distribution. We consider that the benefits of using A-IIT depend on the dimension of the state space and not on the number of modes the state space has.

We don't observe any significant differences in the performance of using IIT or Rejection Free Metropolis Hastings. As the dimension of the state space is small there's no significant advantage to using a different balancing function when choosing a state to move to.



\begin{figure}[htpb]
    \centering
    \includegraphics[width=0.95\linewidth]{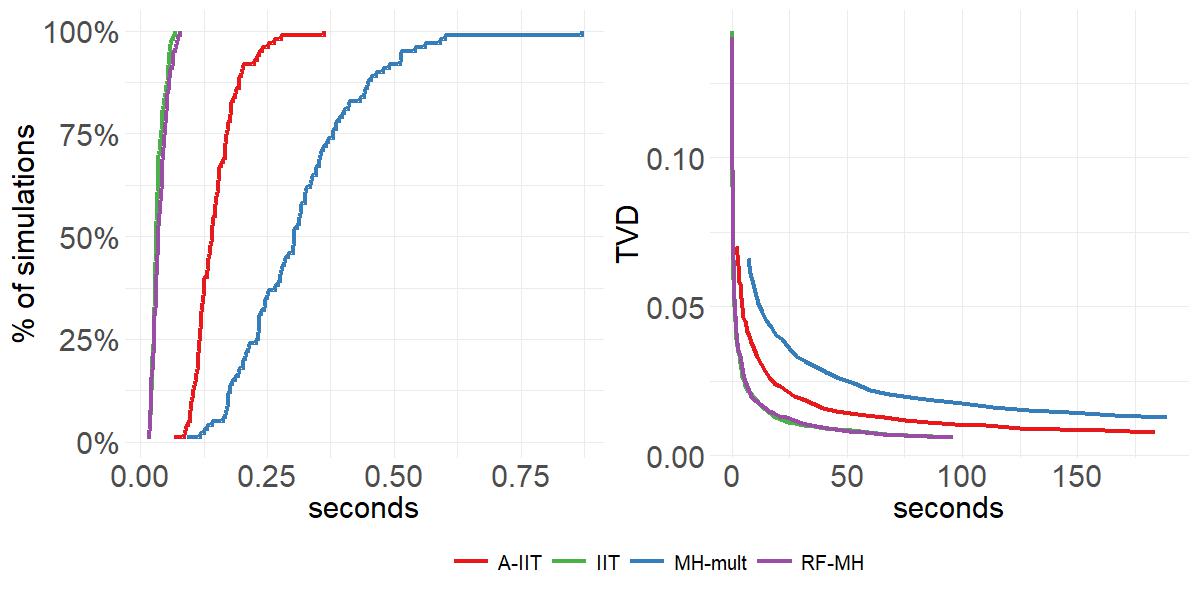}
    \caption{Comparison of performance of the four algorithms in the low dimensional multimodal problem. On the left, time to visit all the modes of the target distribution. On the right, evolution of Total Variation Distance} 
    \label{fig:lowd_multimodal_dual}
\end{figure}

\newpage

\subsection{High dimensional example}

For this example we consider a space of dimension $p=3000$ and we consider 6 modes:

\begin{itemize}
\setlength\itemsep{-0.5em}
    \item $x_{(1)}=(1,0,1,0,\dots,1,0,1,0)$
    \item $x_{(2)}=(0,1,0,1,\dots,0,1,0,1)$
    \item $x_{(3)}=(1,1,1,1,\dots,0,0,0,0)$
    \item $x_{(4)}=(0,0,0,0,\dots,1,1,1,1)$
    \item $x_{(5)}=(0,0,0,0,\dots,0,0,1,1,\dots,1,1,0,0,\dots,0,0,0,0)$
    \item $x_{(6)}=(1,1,1,1,\dots,1,1,0,0,\dots,0,0,1,1 ,\dots,1,1,1,1)$
\end{itemize}

All modes satisfy $||x_{(i)}||_{1}=\frac{p}{2}$ and are defined in such a way that $||x_{(i)}-x_{(j)}||\geq \frac{p}{2}$

The state space is $\cS=\{0,1\}^{3000}$, we use $\theta=0.001$ to define the target distribution.

We run all algorithms for a total of 1 million replica swaps, algorithms (2) and (4) try a replica swap after 2 iterations while the algorithms using a multiplicity list uses an $L_0=800$. For algorithms (1) and (3) 8 out of the 13 chains use the rejection free algorithm and the remaining 5 use the algorithm that allow rejections.  We run 100 simulations of each algorithm to obtain the results presented below.

In table \ref{tab:swap_rate_3k} we see the average number of iterations each replica performs in-between replica swaps with the defined $L_0$. For the 8 replicas using the rejection free algorithm this number corresponds to the average number of rejection-free steps performed before the multiplicity list reaches the value of $L_0$, for the other 5 replicas  the number is fixed as the chain always performs $L_0$ iterations of the algorithm before trying a replica swap.

\begin{table}[!htbp] \centering 
  \caption{Average number of iterations between replica swaps} 
  \label{tab:avg_iter_3k} 
\begin{tabular}{@{\extracolsep{5pt}} ccc} 
\\[-1.8ex]\hline 
\hline \\[-1.8ex] 
replica & A-IIT & MH-mult \\ 
\hline \\[-1.8ex] 
$\beta_{1}$ & 1.614 & 1.162 \\ 
$\beta_{2}$ & 1.809 & 1.352 \\ 
$\beta_{3}$ & 2.001 & 1.788 \\ 
$\beta_{4}$ & 2.262 & 2.587 \\ 
$\beta_{5}$ & 2.678 & 3.916 \\ 
$\beta_{6}$ & 3.235 & 5.346 \\ 
$\beta_{7}$ & 4.052 & 6.816 \\ 
$\beta_{8}$ & 5.209 & 8.588 \\ 
$\beta_{9}$ & 800 & 800 \\ 
$\beta_{10}$ & 800 & 800 \\ 
$\beta_{11}$ & 800 & 800 \\ 
$\beta_{12}$ & 800 & 800 \\ 
$\beta_{13}$ & 800 & 800 \\ 
\hline \\[-1.8ex] 
\end{tabular} 
\end{table}

For algorithms (1) and (3) the burn-in period lasts until the multiplicity list reaches the value 8000. For algorithms (2) and (4) the burn-in period lasts 50 iterations. For algorithm (1) the bounding constant is updated only during the burn-in period. After the burn-in period we stop the adaptation and use the balancing function considering the latest bounding constant found.


\begin{figure}[!h]
    \centering
    \includegraphics[width=0.95\linewidth]{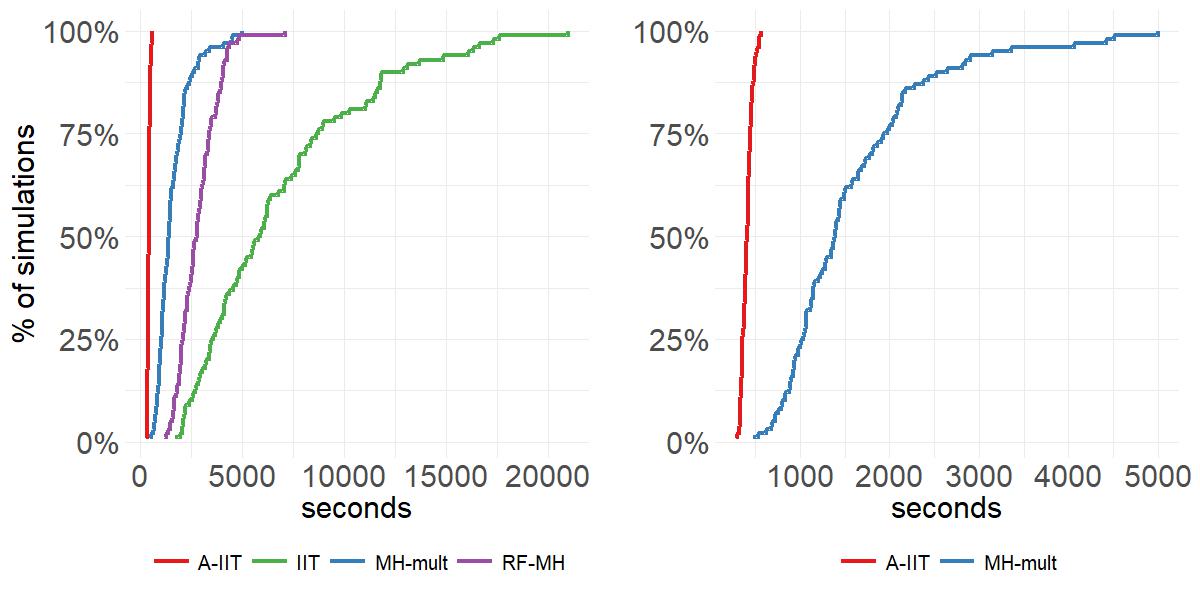}
    \caption{Comparison of time to visit all the modes of the target distribution in a space of dimension $p=3000$. On the right is the same plot only focusin on algorithms A-IIT and MH-mult} 
    \label{fig:highdim_3k_dual}
\end{figure}

The left plot of figure \ref{fig:highdim_3k_dual} presents the time each algorithm took to find all the modes. We can observe that A-IIT and MH-mult, the algorithms that use a multiplicity list, are faster than IIT and RF-MH, the rejection-free algorithms. The time for the algorithms using a multiplicity list could be improved in this case by using SS-IIT in more chains..

The right plot of figure \ref{fig:highdim_3k_dual} we observe the plot only for algorithms (1) and (3). Here we can better appreciate the difference in time. Adaptive IIT is 3 times faster than Rejection Free Metropolis-Hastings. We can attribute the difference in performance to the different balancing function used since that's the main difference between the algorithms.


In figure \ref{fig:highdim_bounds_3k} we can see the plot of the different bounds found for each replica in the simulations. As we can see with the short burn-in period the algorithm was able to find different bounding constants for each replica. When choosing a state for the chain to move algorithm (3) assigns the same probability to all the states with higher probability than the current state, while algorithm (1) may assign different probabilities. This means that algorithm (1) will move quicker to states with high probability and in this setting where the target distribution is multimodal and each mode has a path of high probability leading to it, makes algorithm (1) reach the mode faster.

\begin{figure}[!h]
    \centering
    \includegraphics[width=0.85\linewidth]{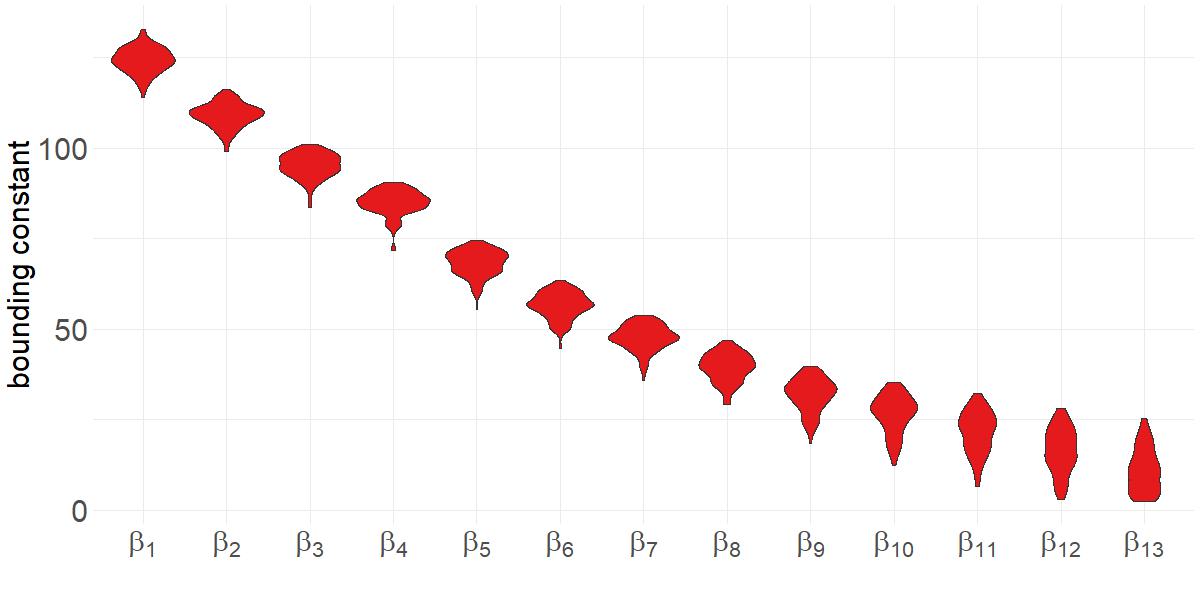}
    \caption{Summary of bounds used for each chain after the burn-in period for the 100 simulations.} 
    \label{fig:highdim_bounds_3k}
\end{figure}

In figure \ref{fig:compare-time-to-mode} we ran similar simulations for algorithms (1) and (3) keeping track of the time it takes for the algorithms to find all the modes but with different number of replicas using the rejection free algorithm. In the x-axis is the number of replicas that used the rejection free algorithm, in the y-axis is the time in seconds it took for the simulations to find the modes.

If we increase the number of replicas using a rejection free algorithm the time to find the modes increases, this is because each iteration requires more computations. However we can see that in all cases A-IIT is faster than MH-mult. This shows the advantage of using $h(r)=\sqrt{r}$ as the balancing function.

\begin{figure}[htbp]
\centering
\includegraphics[width=0.49\textwidth]{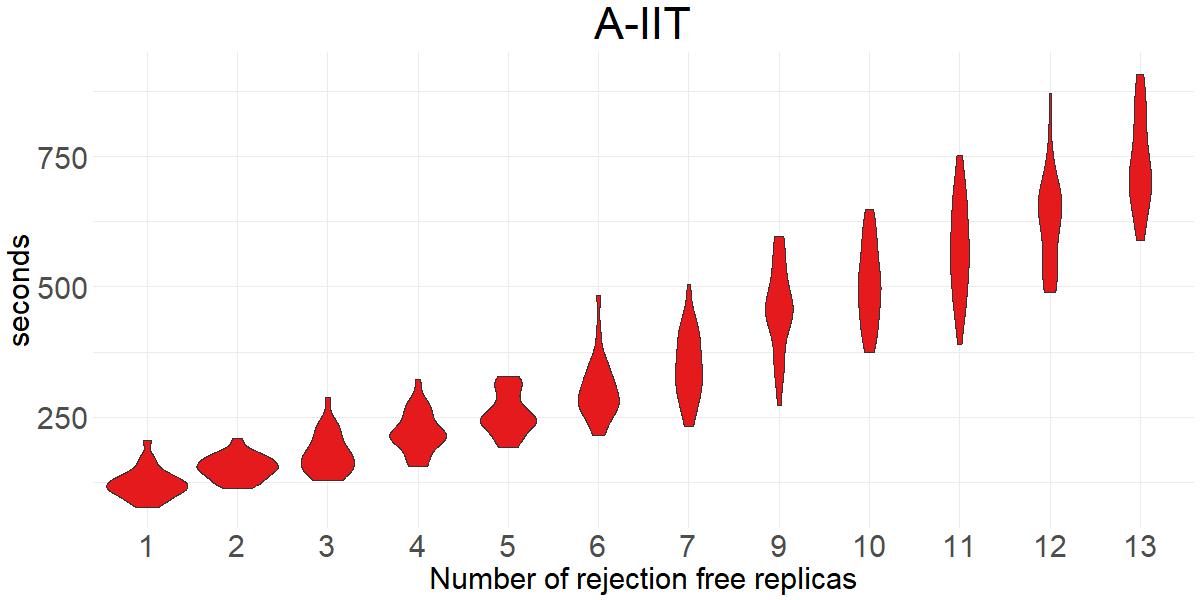}
\hfill
\includegraphics[width=0.49\textwidth]{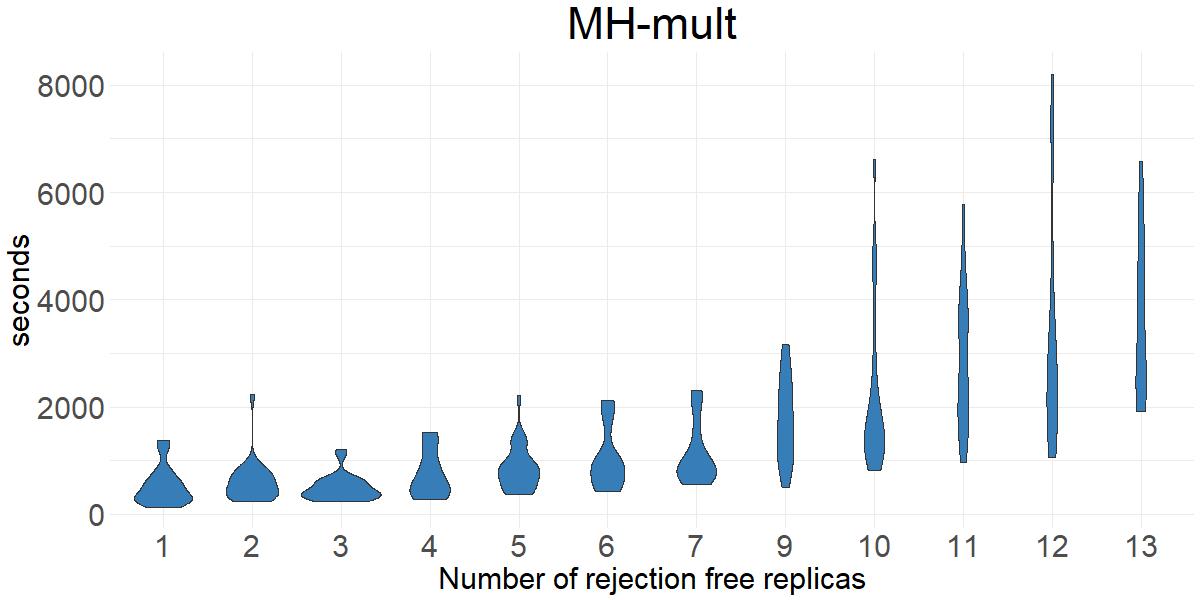}
\caption{Comparison of time to find all modes for different number of rejection free replicas in 100 simulations.}
\label{fig:compare-time-to-mode}
\end{figure}

\newpage
\section{Summary and future work}

This work presented an adaptive MCMC algorithm for discrete state spaces that uses information of the neighbors of the current state to "temper" the choice of the next state to visit. We present both a rejection free version and what we call a single-step version of the algorithm.

We prove the convergence of these algorithms to the desired target distribution, and we show that they are equivalent in the sense that both converge to the same target distribution. This equivalence gives flexibility to the implementation with parallel tempering as we can choose the algorithm that more efficiently uses the available computational resources although the choice will depend on the specific problem to solve.

There are aspects of the algorithm that we don't explore in this work and can be the topic of future research. 
\begin{itemize}
    \item The choice of importance function $f$ to define the bounded balancing function: As was mentioned in the original work of \citet{IIT-Li2023} the choice of balancing functions is problem specific. 
    \item The schedule of adaptations: We mentioned the possibility of reducing the bounding constant towards 1 so we can consider a first stage of constantly updating the constant to improve the exploration of the state space and later reduce it to reduce the variance of the Monte Carlo estimator $\widehat{\E_{\pi}(X)}$
    \item The optimal choice of rejection free or single step replicas: In this work we made a choice on the number of rejection free replicas to use for the simulations trying to use half and half but there's an opportunity to define the optimal number of replicas for each algorithm.
    \item Extensions to other algorithms: In the work of \citet{IIT-Li2023} it's mentioned that IIT allows for other modifications for example to consider a smaller neighbor set at each iteration. We consider that A-IIT can also be adapted in a similar way.
\end{itemize}
\hspace{1cm}

{\small 

\noindent \textbf{Acknowledgments } Thanks to Norma Garcia-Salazar and Marco Gallegos-Herrada for helpful discussions and remarks. This research was enabled in part by support provided by Compute Ontario (computeontario.ca) and the Digital Research Alliance of Canada (alliancecan.ca).

\noindent \textbf{Funding } The authors thank Fujitsu Limited and Fujitsu Consulting (Canada) Inc. for providing financial support.

\noindent \textbf{Conflict of Interest } The authors declare no competing interests.
}

\newpage
{\small
\bibliographystyle{apacite}
\bibliography{references}
} 
\appendix
\section*{Appendix}
\renewcommand{\thesubsection}{\Alph{subsection}}

\subsection{Details on the simulations}\label{app:sim_details}
Here we present more details on the simulations. Specifically the values of $\beta$ used, the average swap rate and we present other results obtained in spaces of dimension 1000, 5000 and 7000.
\newpage
\subsubsection{Low dimensional examples}

\begin{table}[!htbp] \centering 
  \caption{Inverse temperature used in the bimodal model} 
  \label{tab:temperatures_bimodal} 
\begin{tabular}{@{\extracolsep{5pt}} ccccc} 
\\[-1.8ex]\hline 
\hline \\[-1.8ex] 
temp\_id & A-IIT & MH-mult & IIT & RF-MH \\ 
\hline \\[-1.8ex] 
$\beta_{1}$ & 1 & 1 & 1 & 1 \\ 
$\beta_{2}$ & 0.49 & 0.49 & 0.26 & 0.26 \\ 
$\beta_{3}$ & 0.33 & 0.33 & 0.06 & 0.06 \\ 
$\beta_{4}$ & 0.22 & 0.22 & NA & NA \\ 
\hline \\[-1.8ex] 
\end{tabular} 
\end{table}

\begin{table}[!htbp] \centering 
  \caption{Average swap rate between replica $\beta_{i}$ and $\beta_{i+1}$ for the bimodal model} 
  \label{tab:swap_rate_bimodal} 
\begin{tabular}{@{\extracolsep{5pt}} ccccc} 
\\[-1.8ex]\hline 
\hline \\[-1.8ex] 
replica & A-IIT & IIT & MH-mult & RF-MH \\ 
\hline \\[-1.8ex] 
$\beta_{1}$ & 0.2352 & 0.2382 & 0.2352 & 0.2382 \\ 
$\beta_{2}$ & 0.2423 & 0.2327 & 0.2423 & 0.2326 \\ 
$\beta_{3}$ & 0.2521 & NA & 0.2521 & NA \\ 
\hline \\[-1.8ex] 
\end{tabular} 
\end{table}

\begin{table}[!htbp] \centering 
  \caption{Inverse temperature used in the 7 modes model} 
  \label{tab:temperatures_7_mode} 
\begin{tabular}{@{\extracolsep{5pt}} ccccc} 
\\[-1.8ex]\hline 
\hline \\[-1.8ex] 
temp\_id & A-IIT & MH-mult & IIT & RF-MH \\ 
\hline \\[-1.8ex] 
$\beta_{1}$ & 1 & 1 & 1 & 1 \\ 
$\beta_{2}$ & 0.31 & 0.31 & 0.15 & 0.155 \\ 
$\beta_{3}$ & 0.21 & 0.21 & 0.002 & 0.002 \\ 
\hline \\[-1.8ex] 
\end{tabular} 
\end{table}

\begin{table}[!htbp] \centering 
  \caption{Average swap rate between replica $\beta_{i}$ and $\beta_{i+1}$ for the 7 modes model} 
  \label{tab:swap_rate_7_mode} 
\begin{tabular}{@{\extracolsep{5pt}} ccccc} 
\\[-1.8ex]\hline 
\hline \\[-1.8ex] 
replica & A-IIT & IIT & MH-mult & RF-MH \\ 
\hline \\[-1.8ex] 
$\beta_{1}$ & 0.2478 & 0.2327 & 0.2477 & 0.2506 \\ 
$\beta_{2}$ & 0.2521 & 0.2539 & 0.2521 & 0.2391 \\ 
\hline \\[-1.8ex] 
\end{tabular} 
\end{table}

\newpage
\subsubsection{High dimensional example}

\begin{table}[!htbp] \centering 
  \caption{Inverse temperature used in the problem of dimension 3k} 
  \label{tab:temperatures_3k} 
\begin{tabular}{@{\extracolsep{5pt}} ccccc} 
\\[-1.8ex]\hline 
\hline \\[-1.8ex] 
temp\_id & A-IIT & IIT & MH-mult & RF-MH \\ 
\hline \\[-1.8ex] 
$\beta_{1}$ & 20000 & 20000 & 20000 & 20000 \\ 
$\beta_{2}$ & 19517 & 15046 & 17899 & 15005 \\ 
$\beta_{3}$ & 19029 & 13509 & 15895 & 13611 \\ 
$\beta_{4}$ & 18535 & 12162 & 14353 & 12684 \\ 
$\beta_{5}$ & 17744 & 11215 & 13057 & 12182 \\ 
$\beta_{6}$ & 17125 & 10570 & 12234 & 11600 \\ 
$\beta_{7}$ & 16568 & 10087 & 11631 & 11377 \\ 
$\beta_{8}$ & 16037 & 9716 & 11093 & 11185 \\ 
$\beta_{9}$ & 15571 & 9396 & 10578 & 11090 \\ 
$\beta_{10}$ & 15273 & 9166 & 10109 & 10986 \\ 
$\beta_{11}$ & 15075 & 9001 & 9409 & 10892 \\ 
$\beta_{12}$ & 14786 & 8827 & 8951 & 10802 \\ 
$\beta_{13}$ & 14595 & 8691 & 8417 & 10735 \\ 
\hline \\[-1.8ex] 
\end{tabular} 
\end{table}

\begin{table}[!htbp] \centering 
  \caption{Average swap rate between replica $\beta_{i}$ and $\beta_{i+1}$ in space of dimension 3000} 
  \label{tab:swap_rate_3k} 
\begin{tabular}{@{\extracolsep{5pt}} ccccc} 
\\[-1.8ex]\hline 
\hline \\[-1.8ex] 
replica & A-IIT & IIT & MH-mult & RF-MH \\ 
\hline \\[-1.8ex] 
$\beta_{1}$ & 0.4489 & 0.2509 & 0.3759 & 0.2161 \\ 
$\beta_{2}$ & 0.4582 & 0.342 & 0.2893 & 0.3295 \\ 
$\beta_{3}$ & 0.4333 & 0.2239 & 0.2381 & 0.3212 \\ 
$\beta_{4}$ & 0.4009 & 0.222 & 0.1943 & 0.4041 \\ 
$\beta_{5}$ & 0.3953 & 0.2743 & 0.2374 & 0.2964 \\ 
$\beta_{6}$ & 0.3859 & 0.3078 & 0.2604 & 0.4368 \\ 
$\beta_{7}$ & 0.3664 & 0.3336 & 0.2586 & 0.4125 \\ 
$\beta_{8}$ & 0.3505 & 0.3144 & 0.2383 & 0.4716 \\ 
$\beta_{9}$ & 0.3575 & 0.3507 & 0.2287 & 0.386 \\ 
$\beta_{10}$ & 0.3942 & 0.4006 & 0.0984 & 0.3795 \\ 
$\beta_{11}$ & 0.3503 & 0.3268 & 0.1639 & 0.3212 \\ 
$\beta_{12}$ & 0.3703 & 0.4098 & 0.0946 & 0.4413 \\ 
\hline \\[-1.8ex] 
\end{tabular} 
\end{table} 

\newpage
\subsection{Other results in high dimensions}\label{app:other_resuts}
We chose to present the results of the simulations of a problem in a state space of dimension 3000 as this showed clearly the advantages of A-IIT over other algorithms. Here we share some other results obtained from simulations run with similar target distributions $\pi$ with the same number of modes but in other dimensions.

Although the results presented for the state space of dimension 3000 are not replicated in dimension 5000, we observe similar results for dimension 7000. We consider there is an opportunity to optimize some parameters in dimension 5000 so that we can show that A-IIT performs better than Metropolis Hastings

\subsubsection{Dimension 1000}

\begin{figure}[!h]
    \centering
    \includegraphics[width=0.85\linewidth]{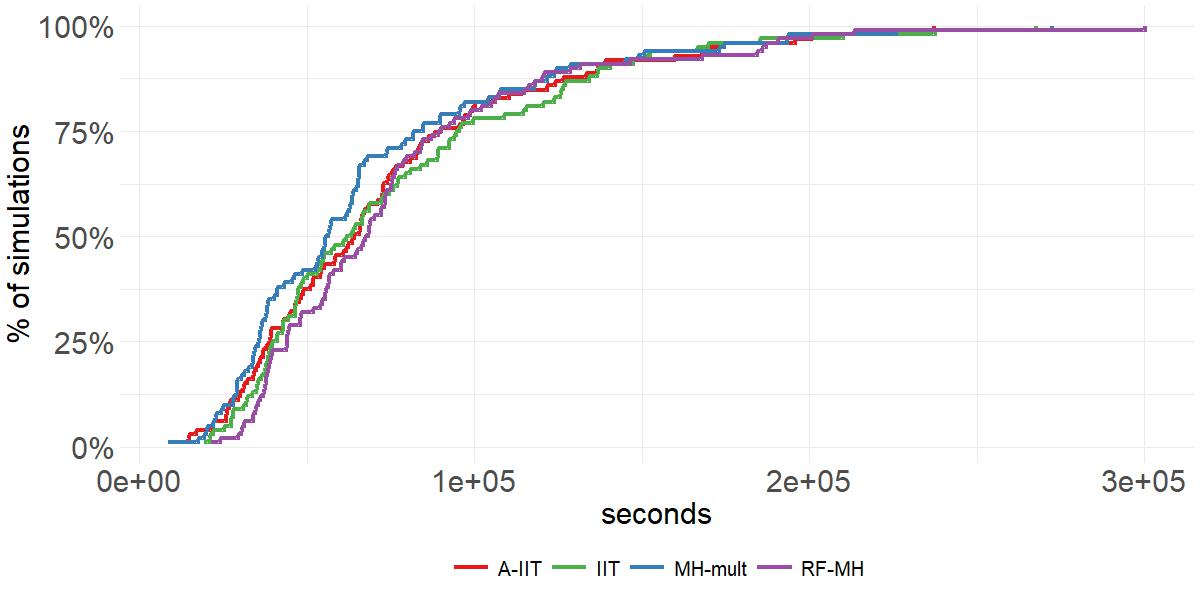}
    \caption{Comparison of time to visit all the modes of the target distribution in a space of dimension $p=1000$.} 
    \label{fig:highdim_speed_mode_1k}
\end{figure}

\begin{table}[!htbp] \centering 
  \caption{Inverse temperature used in the problem of dimension 1k} 
  \label{tab:temperatures_1k} 
\begin{tabular}{@{\extracolsep{5pt}} ccccc} 
\\[-1.8ex]\hline 
\hline \\[-1.8ex] 
temp\_id & A-IIT & IIT & MH-mult & RF-MH \\ 
\hline \\[-1.8ex] 
$\beta_{1}$ & 20000 & 20000 & 20000 & 20000 \\ 
$\beta_{2}$ & 19209 & 18092 & 19175 & 18202 \\ 
$\beta_{3}$ & 18544 & 16632 & 18496 & 16829 \\ 
$\beta_{4}$ & 18030 & 15661 & 17962 & 15769 \\ 
$\beta_{5}$ & 17441 & 14723 & 17333 & 14902 \\ 
$\beta_{6}$ & 16904 & 13981 & 16728 & 14075 \\ 
$\beta_{7}$ & 16480 & 13544 & 16322 & 13582 \\ 
$\beta_{8}$ & 15969 & 12819 & 15876 & 12999 \\ 
$\beta_{9}$ & 15483 & 12279 & 15397 & 12419 \\ 
$\beta_{10}$ & 15115 & 11806 & 15052 & 11994 \\ 
$\beta_{11}$ & 14768 & 11372 & 14716 & 11472 \\ 
$\beta_{12}$ & 14376 & 10970 & 14393 & 11043 \\ 
$\beta_{13}$ & 13910 & 10484 & 13968 & 10691 \\ 
\hline \\[-1.8ex] 
\end{tabular} 
\end{table}

\begin{table}[!htbp] \centering 
  \caption{Average swap rate between replica $\beta_{i}$ and $\beta_{i+1}$ in space of dimension 1000} 
  \label{tab:swap_rate_1k} 
\begin{tabular}{@{\extracolsep{5pt}} ccccc} 
\\[-1.8ex]\hline 
\hline \\[-1.8ex] 
replica & A-IIT & IIT & MH-mult & RF-MH \\ 
\hline \\[-1.8ex] 
$\beta_{1}$ & 0.333 & 0.2655 & 0.3258 & 0.2971 \\ 
$\beta_{2}$ & 0.345 & 0.2824 & 0.3408 & 0.3201 \\ 
$\beta_{3}$ & 0.369 & 0.3951 & 0.3631 & 0.3636 \\ 
$\beta_{4}$ & 0.3397 & 0.3454 & 0.3275 & 0.3922 \\ 
$\beta_{5}$ & 0.3417 & 0.3968 & 0.3199 & 0.3537 \\ 
$\beta_{6}$ & 0.3651 & 0.5849 & 0.3678 & 0.5423 \\ 
$\beta_{7}$ & 0.3287 & 0.3205 & 0.3468 & 0.4315 \\ 
$\beta_{8}$ & 0.3254 & 0.4131 & 0.3255 & 0.3909 \\ 
$\beta_{9}$ & 0.3576 & 0.4338 & 0.3648 & 0.4932 \\ 
$\beta_{10}$ & 0.3583 & 0.4378 & 0.3617 & 0.3612 \\ 
$\beta_{11}$ & 0.3324 & 0.4397 & 0.3603 & 0.4164 \\ 
$\beta_{12}$ & 0.2914 & 0.3085 & 0.3092 & 0.4715 \\ 
\hline \\[-1.8ex] 
\end{tabular} 
\end{table} 

\newpage
\subsubsection{Dimension 5000}

\begin{figure}[htpb]
    \centering
    \includegraphics[width=0.95\linewidth]{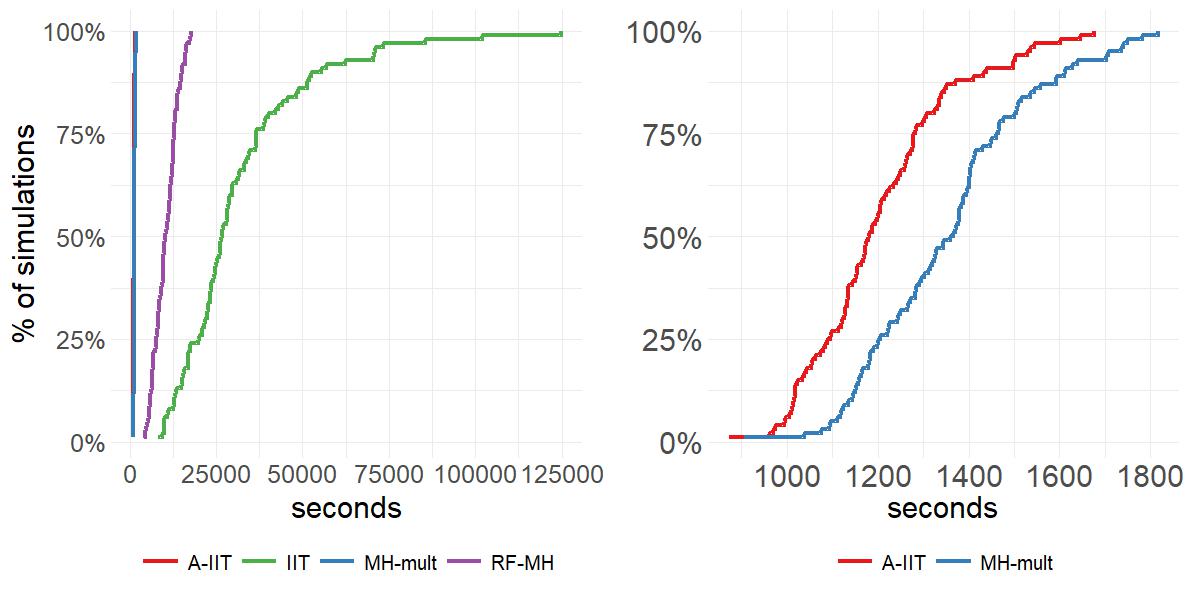}
    \caption{Comparison of time to visit all the modes of the target distribution in a space of dimension $p=5000$. On the right is the same plot only focusin on A-IIT and MH-mult.} 
    \label{fig:highdim_5k_dual}
\end{figure}

\begin{table}[!htbp] \centering 
  \caption{Inverse temperature used in the problem of dimension 5k} 
  \label{tab:temperatures_5k} 
\begin{tabular}{@{\extracolsep{5pt}} ccccc} 
\\[-1.8ex]\hline 
\hline \\[-1.8ex] 
temp\_id & A-IIT & IIT & MH-mult & RF-MH \\ 
\hline \\[-1.8ex] 
$\beta_{1}$ & 20000 & 20000 & 20000 & 20000 \\ 
$\beta_{2}$ & 19813 & 11152 & 19199 & 10929 \\ 
$\beta_{3}$ & 19524 & 9798 & 18256 & 9985 \\ 
$\beta_{4}$ & 19387 & 8698 & 17010 & 9369 \\ 
$\beta_{5}$ & 19254 & 7903 & 16246 & 8908 \\ 
$\beta_{6}$ & 19135 & 7290 & 15526 & 8538 \\ 
$\beta_{7}$ & 19075 & 6819 & 14787 & 8233 \\ 
$\beta_{8}$ & 18891 & 6610 & 14075 & 7963 \\ 
$\beta_{9}$ & 18680 & 6427 & 13404 & 7723 \\ 
$\beta_{10}$ & 18616 & 6274 & 12786 & 7523 \\ 
$\beta_{11}$ & 18536 & 6125 & 12388 & 7422 \\ 
$\beta_{12}$ & 18454 & 6007 & 11783 & 7299 \\ 
$\beta_{13}$ & 18421 & 5889 & 11230 & 7208 \\ 
\hline \\[-1.8ex] 
\end{tabular} 
\end{table}

\begin{table}[!htbp] \centering 
  \caption{Average swap rate between replica $\beta_{i}$ and $\beta_{i+1}$ in space of dimension 5000} 
  \label{tab:swap_rate_5k} 
\begin{tabular}{@{\extracolsep{5pt}} ccccc} 
\\[-1.8ex]\hline 
\hline \\[-1.8ex] 
replica & A-IIT & IIT & MH-mult & RF-MH \\ 
\hline \\[-1.8ex] 
$\beta_{1}$ & 0.4168 & 0.352 & 0.3636 & 0.2386 \\ 
$\beta_{2}$ & 0.3777 & 0.3035 & 0.3773 & 0.3558 \\ 
$\beta_{3}$ & 0.4079 & 0.171 & 0.3478 & 0.3633 \\ 
$\beta_{4}$ & 0.3869 & 0.1459 & 0.3838 & 0.3457 \\ 
$\beta_{5}$ & 0.3624 & 0.1368 & 0.3722 & 0.3313 \\ 
$\beta_{6}$ & 0.3943 & 0.147 & 0.3776 & 0.3219 \\ 
$\beta_{7}$ & 0.3077 & 0.3699 & 0.3593 & 0.269 \\ 
$\beta_{8}$ & 0.2868 & 0.3774 & 0.3591 & 0.2445 \\ 
$\beta_{9}$ & 0.3337 & 0.3816 & 0.3461 & 0.2021 \\ 
$\beta_{10}$ & 0.3391 & 0.3427 & 0.3786 & 0.244 \\ 
$\beta_{11}$ & 0.3133 & 0.3729 & 0.306 & 0.1981 \\ 
$\beta_{12}$ & 0.3232 & 0.3278 & 0.294 & 0.2217 \\ 
\hline \\[-1.8ex] 
\end{tabular} 
\end{table}

\newpage
\subsubsection{Dimension 7000}

\begin{figure}[htpb]
    \centering
    \includegraphics[width=0.9\linewidth]{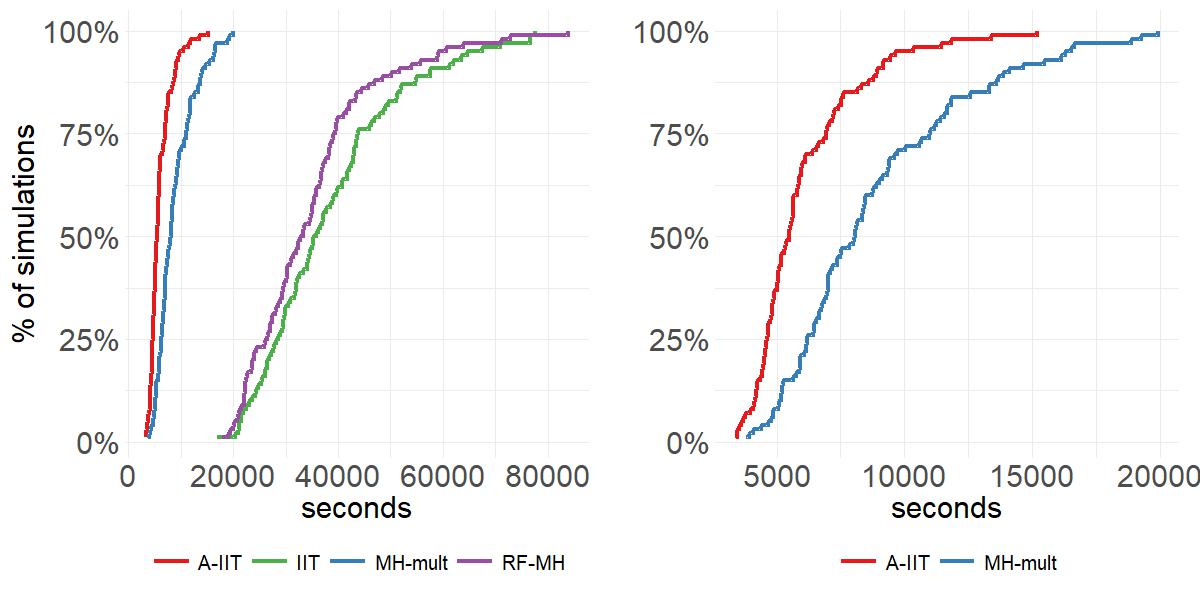}
    \caption{Comparison of time to visit all the modes of the target distribution in a space of dimension $p=7000$. On the right is the same plot only focus in on A-IIT and MH-mult.} 
    \label{fig:highdim_7k_dual}
\end{figure}

\begin{table}[!htbp] \centering 
  \caption{Inverse temperature used in the problem of dimension 7k} 
  \label{tab:temperatures_7k} 
\begin{tabular}{@{\extracolsep{5pt}} ccccc} 
\\[-1.8ex]\hline 
\hline \\[-1.8ex] 
temp\_id & A-IIT & IIT & MH-mult & RF-MH \\ 
\hline \\[-1.8ex] 
$\beta_{1}$ & 20000 & 20000 & 20000 & 20000 \\ 
$\beta_{2}$ & 11710 & 10008 & 11457 & 9951 \\ 
$\beta_{3}$ & 9796 & 8637 & 9845 & 8781 \\ 
$\beta_{4}$ & 9154 & 7937 & 9005 & 8063 \\ 
$\beta_{5}$ & 8692 & 7426 & 8348 & 7557 \\ 
$\beta_{6}$ & 8327 & 7024 & 7813 & 7135 \\ 
$\beta_{7}$ & 8023 & 6751 & 7377 & 6870 \\ 
$\beta_{8}$ & 7738 & 6521 & 7049 & 6649 \\ 
$\beta_{9}$ & 7490 & 6310 & 6741 & 6483 \\ 
$\beta_{10}$ & 7274 & 6185 & 6472 & 6371 \\ 
$\beta_{11}$ & 7069 & 6110 & 6239 & 6301 \\ 
$\beta_{12}$ & 6878 & 6032 & 6029 & 6252 \\ 
$\beta_{13}$ & 6713 & 5996 & 5858 & 6214 \\ 
\hline \\[-1.8ex] 
\end{tabular} 
\end{table}

\begin{table}[!htbp] \centering 
  \caption{Average swap rate between replica $\beta_{i}$ and $\beta_{i+1}$ in space of dimension 7000} 
  \label{tab:swap_rate_7k} 
\begin{tabular}{@{\extracolsep{5pt}} ccccc} 
\\[-1.8ex]\hline 
\hline \\[-1.8ex] 
replica & A-IIT & IIT & MH-mult & RF-MH \\ 
\hline \\[-1.8ex] 
$\beta_{1}$ & 0.3965 & 0.363 & 0.3854 & 0.336 \\ 
$\beta_{2}$ & 0.2348 & 0.206 & 0.2628 & 0.256 \\ 
$\beta_{3}$ & 0.3199 & 0.2414 & 0.2614 & 0.2394 \\ 
$\beta_{4}$ & 0.3341 & 0.2253 & 0.2418 & 0.2287 \\ 
$\beta_{5}$ & 0.329 & 0.217 & 0.2207 & 0.1984 \\ 
$\beta_{6}$ & 0.3428 & 0.2423 & 0.2226 & 0.2348 \\ 
$\beta_{7}$ & 0.3314 & 0.2242 & 0.2412 & 0.2136 \\ 
$\beta_{8}$ & 0.3336 & 0.1926 & 0.2274 & 0.2091 \\ 
$\beta_{9}$ & 0.3264 & 0.2245 & 0.2272 & 0.2153 \\ 
$\beta_{10}$ & 0.3182 & 0.257 & 0.2362 & 0.2371 \\ 
$\beta_{11}$ & 0.3081 & 0.221 & 0.23 & 0.233 \\ 
$\beta_{12}$ & 0.3226 & 0.295 & 0.2623 & 0.2666 \\ 
\hline \\[-1.8ex] 
\end{tabular} 
\end{table}

\newpage 
\subsection{Algorithms}\label{app:algorithms}

\begin{algorithm}
\caption{Metropolis Hastings with a multiplicity list}
\label{alg:rejection-free-metropolis-hastings-mult}
\begin{algorithmic}
\State initialize $X_0$
\For{$k$ in $1$ to $K$}
\State compute
    $Z(X_{k-1}) = \sum_{Y \in \mathcal{N}_{X_{k-1}}}\mathcal{Q}(Y|X_{k-1})\min\Bigg\{1, \frac{\pi(Y)\mathcal{Q}(X_{k-1}|Y)}{\pi(X_{k-1})\mathcal{Q}(Y|X_{k-1})}\Bigg\}$
\State simulate multiplicity list $M_{k-1} \gets 1 + G$ where $G \sim \text{Geometric}(Z(X_{k-1}))$ with 
\State Set $\color{blue}w_{k-1}(X_{k-1})\gets M_{k-1}$
    \State choose the next state $X_{k} \in \mathcal{N}_{X_{k-1}}$ such that 
    {
\abovedisplayskip=-2pt
\belowdisplayskip=-2pt
$$\widehat{\mathbf{P}}(X_{k} = Y \mid X_{k-1}) \propto \mathcal{Q}(Y|X_{k-1})\min\Bigg\{1, \frac{\pi(Y) \mathcal{Q}(X_{k-1}|Y)}{\pi(X_{k-1}) \mathcal{Q}(Y|X_{k-1})}\Bigg\}$$
}
    
\EndFor
\end{algorithmic}
\end{algorithm}

\begin{algorithm}
\caption{Rejection-Free Metropolis Hastings}
\label{alg:rejection-free-metropolis-hastings-weight}
\begin{algorithmic}
\State initialize $X_0$
\For{$k$ in $1$ to $K$}
\State compute
    $Z(X_{k-1}) = \sum_{Y \in \mathcal{N}_{X_{k-1}}}\mathcal{Q}(Y|X_{k-1})\min\Bigg\{1, \frac{\pi(Y)\mathcal{Q}(X_{k-1}|Y)}{\pi(X_{k-1})\mathcal{Q}(Y|X_{k-1})}\Bigg\}$
\State Set $\color{blue} w_{k-1}(X_{k-1})\gets \frac{1}{Z(X_{k-1})}$
    \State choose the next state $X_{k} \in \mathcal{N}_{X_{k-1}}$ such that 
    {
\abovedisplayskip=-2pt
\belowdisplayskip=-2pt
    $$\widehat{\mathbf{P}}(X_{k} = Y \mid X_{k-1}) \propto \mathcal{Q}(Y|X_{k-1})\min\Bigg\{1, \frac{\pi(Y) \mathcal{Q}(X_{k-1}|Y)}{\pi(X_{k-1}) \mathcal{Q}(Y|X_{k-1})}\Bigg\}$$
    }
\EndFor
\end{algorithmic}
\end{algorithm}

\begin{algorithm}
\caption{Metropolis-Hastings algorithm}\label{alg:metropolis-hastings}
\begin{algorithmic}
\State initialize $X_0$
\For{$k$ in $1$ to $K$}
    \State random $Y \in \mathcal{N}_{X_{k-1}}$ based on $\mathcal{Q}(\cdot|X_{k-1})$
    \State Set $\color{blue} w_{k-1}(X_{k-1})\gets 1$
    \State $X_{k} \gets Y$ with probability $\frac{\pi(Y) \mathcal{Q}(Y, X_{k-1})}{\pi(X_{k-1}) \mathcal{Q}(X_{k-1}, Y)}$ \Comment{Accept proposed move to $Y$}
    \State Otherwise $X_{k} \gets X_{k-1}$    \Comment{Reject and stay at $X_{k-1}$}

\EndFor
\end{algorithmic}
\end{algorithm}

\end{document}